\documentclass[11pt]{article}
\usepackage[utf8]{inputenc}
\usepackage{amsfonts,latexsym,amsmath,amsthm,amssymb,amscd,euscript,mathrsfs,graphicx}
\usepackage{framed}
\usepackage{fullpage}
\usepackage{color}
\usepackage[shortlabels]{enumitem}
\usepackage[obeyFinal,textsize=scriptsize,shadow]{todonotes}
\usepackage{tikz}
\usetikzlibrary{matrix}
\usepackage{hyperref}
\usepackage{marginnote}
\usepackage{algorithm}
\usepackage{algpseudocode}

\usepackage{multirow} 

\newtheorem{theorem}{Theorem}
\newtheorem{proposition}[theorem]{Proposition}
\newtheorem{lemma}[theorem]{Lemma}
\newtheorem{corollary}[theorem]{Corollary}
\theoremstyle{definition}
\newtheorem{definition}[theorem]{Definition}

\theoremstyle{remark}
\newtheorem*{remark}{Remark}

\newcommand{\BC}{\mathbb C}
\newcommand{\BE}{\mathbb E}

\newcommand{\BP}{\mathbb P}

\newcommand{\BR}{\mathbb R}

\renewcommand{\epsilon}{\varepsilon}
\newcommand{\eps}{\varepsilon}

\title{Tight and Robust Private Mean Estimation with Few Users}
\author{Hossein Esfandiari\thanks{Google Research. Email: \texttt{esfandiari@google.com}} , Vahab Mirrokni\thanks{Google Research. Email: \texttt{mirrokni@google.com}} , Shyam Narayanan\thanks{MIT. Work done as an intern at Google Research. Email: \texttt{shyamsn@mit.edu}}}
\date{\today}

\begin{document}

\maketitle

\begin{abstract}
In this work, we study high-dimensional mean estimation under user-level differential privacy, and design an $(\varepsilon,\delta)$-differentially private mechanism using as few users as possible. In particular, we provide a nearly optimal trade-off between the number of users and the number of samples per user required for private mean estimation, even when the number of users is as low as $O(\frac{1}{\varepsilon}\log\frac{1}{\delta})$. Interestingly, this bound on the number of \emph{users} is independent of the dimension (though the number of \emph{samples per user} is allowed to depend polynomially on the dimension), unlike the previous work that requires the number of users to depend polynomially on the dimension. This resolves a problem first proposed by Amin et al.~\cite{amin2019biasvarianceuser}. 
Moreover, our mechanism is robust against corruptions in up to $49\%$ of the users.
Finally, our results also apply to optimal algorithms for privately learning discrete distributions with few users, answering a question of Liu et al.~\cite{liu2020discrete}, and a broader range of problems such as stochastic convex optimization and a variant of stochastic gradient descent via a reduction to differentially private mean estimation.
\end{abstract}

\section{Introduction} \label{sec:intro}
Without any doubt, user privacy is one of the greatest matters of the modern era. Differential privacy is one of the most popular concepts to measure and enforce privacy. In a nutshell, differential privacy statistically protects the privacy of every single unit of the data, regardless of how much information about the rest of the data is revealed. Depending on how a single unit of data is defined, we have two general notions of differential privacy: user-level privacy and item-level~\footnote{Also known as record-level or sample-level differential privacy.} privacy. 

Let us differentiate these two notions with a simple example. Consider a dataset that consists of the daily blood pressures of the patients in a hospital. For the privacy purpose, we may define a unit of data to be one single value of a blood pressure, or the whole collection of the blood pressures of a patient. The former leads to item-level differential privacy, while the latter leads to user-level differential privacy. Roughly speaking, item-level differential privacy in this example protects one single value of a patient's blood pressure; however, an estimate of the average blood pressure of an individual might be released under item-level differential privacy, which is concerning in this example. Conversely, user-level differential privacy protects any information about the blood pressure of every single patient, while it allows statistical analysis of a population of the patients. 
In general user-level differential privacy provides a higher level of privacy. User-level privacy is certainly necessary in some cases such as our previous example because an individual's average blood pressure is clearly a private information that should be kept protected.
Hence, in this work, we mainly focus on user-level differential privacy.

Ignoring the computational constraints, it is folklore that providing differential privacy gets easier as the number of users in the dataset grows and for many tasks we can have simple differentially private mechanisms using a \emph{large enough} number of users in the dataset. However, the question has always been \emph{``how large should the dataset be?''}. Providing differential privacy using only a few users is of particular importance when we want to analyse local information, for example analysing the information of each hospital separately, or analysing the information of each neighborhood of a city separately (e.g. for advertising an anti-drug campaign or suicide prevention lifeline). Another example where we face few users is when we are analysing a rare situation such as the patients of a rare or novel disease (e.g. the Covid-19 pandemic at its early days~\cite{corsello2020impact}).
 
Motivated by this, we consider the problem of mean estimation with user-level differential privacy and attempt to understand the limits of the number of users one needs to provide a differentially private mechanism for this problem. Indeed, differentially private mean estimation is of high technical and practical importance since it can be used directly to develop differentially private mechanisms for several other tasks such as stochastic convex optimization, empirical risk minimization, and a variant of stochastic gradient descent~\cite{levy2021user}. 

\subsection{Problem Formulation and Results} \label{subsec:results}
We start with a formal definition of differential privacy. First, let $X = \{X_{i, j}\}$ represent a dataset, where each user $i$ (ranging from $1$ to $n$) has $m$ samples ($j$ ranging from $1$ to $m$). In the context of user-level privacy, we say that two datasets $X = \{X_{i, j}\}$ and $X' = \{X'_{i, j}\}$ are \emph{neighboring} if there is only one user $i$ where the datasets differ, i.e., for all $i' \neq i$ and all $j \in [m]$, $X_{i', j} = X'_{i', j}$. In the context of item-level privacy, the definition of neighborhood is more restrictive. Specifically, in item-level differential privacy two neighboring datasets only differ in a single sample (i.e., for only one pair $(i, j)$ we can have $X_{i, j} \neq X'_{i, j}$).

\begin{definition} \label{def}
     We say that a (randomized) algorithm $M$ is $(\eps, \delta)$-differentially private if for any two neighboring datasets $X$ and $X'$ and any subset $S$ of the output space,
\[\BP[M(X) \in S] \le e^{\eps} \cdot \BP[M(X') \in S] + \delta.\]
\end{definition}

As noted before, Definition \ref{def} applies to both user-level and item-level differential privacy, while their only difference is in their definition of neighborhood.
Although in this work we focus on user-level privacy, we first study a case where each user has one sample (i.e., $m=1$), in which user-level privacy and item-level privacy are equivalent.


Our main contribution in this work is to provide a nearly optimal trade-off between the number of users $n$ and the number of samples $m$ per user required so that one can accurately and privately learn the mean of a distribution with a low error. Although such a trade-off was known when the number of users $n$ is sufficiently large, it was poorly understood when $n$ is small. 

Similar to the previous work on user-level differential privacy~\cite{liu2020discrete,levy2021user}, we assume that each sample is drawn independently from some unknown distribution $\mathcal{D}$, with a support entirely in an unknown ball $B \subset \BR^d$ with center $x \in \BR^d$ and radius $r$. For simplicity, we suppose $r=1$ in this subsection. 

It is worth noting that our algorithms enjoy a high level of robustness. Specifically, we show that our algorithm still succeeds even if the samples of each user are mildly correlated, or even if the data of $49\%$ of the users are adversarially corrupted. It is easy to observe that it is impossible to recover the mean to any degree if the data of $51\%$ of the users are adversarially corrupted.

Next, we describe our contributions in more detail.
First, we provide an $(\eps, \delta)$ user-level differentially private algorithm over $n \ge O(\frac{1}{\eps} \log \frac{1}{\delta})$ users, that estimates the mean $\mu$ of the unknown distribution $\mathcal{D}$ up to error $O(\sqrt{d/m})$ (see Theorem \ref{thm:user_main} for details). Hence, in the infinite sample limit (i.e., when each user is given $m \to \infty$ i.i.d. samples), we can estimate the mean $\mu$ perfectly with very high probability as long as we have $n \ge \frac{1}{\eps} \log \frac{1}{\delta}$ users.
We complement this positive result with a matching lower bound that shows, even if there are infinitely many samples per user, it is necessary to have $\Omega(\frac{1}{\eps} \log \frac{1}{\delta})$ users in order to privately estimate $\mu$ up to any finite accuracy (see Lemma \ref{lem:lb_easy} for details).
Together, these fully resolve the question of how many users are needed to estimate the mean in the infinite sample limit, a question first investigated by Amin et al.~\cite{amin2019biasvarianceuser} and then further studied by Levy et al.~\cite{levy2021user}. Previously, this was only established when the number of users $n$ was at least $\tilde{O}(\sqrt{d \log \frac{1}{\delta}}/\eps)$
\cite{levy2021user}.

Next, we obtain nearly tight bounds for the user-sample trade-off needed for private mean estimation. Specifically, when we have $n \ge O(\frac{1}{\eps} \log \frac{1}{\delta})$ users and $m$ samples per user, we provide an $(\eps, \delta)$ user-level differentially private algorithm that estimates the mean $\mu$ up to $\ell_2$ error $\tilde{O}\left(\sqrt{\frac{1}{mn} + \frac{d}{mn^2 \eps^2}}\right)$ with probability $1-\alpha$. It is worth noting that the running time of our algorithm can be made almost linear in $n, m, d$, and $\log \frac{1}{\alpha}$ (see Theorem \ref{thm:user_main} for details). This extends the algorithm of Levy et al.~\cite{levy2021user}, which provided similar bounds but only when $n \ge \tilde{O}(\frac{1}{\eps} \sqrt{d \log \frac{1}{\delta}})$, to the regime when $n$ can be very low with no dependency on the dimension $d$. Moreover, our $\ell_2$ error is known to be tight up to logarithmic factors when $\frac{1}{\delta}$ is a sufficiently large polynomial in the dimension $d$~\cite{kamath2019highdimensional,levy2021user}, implying that the error of our algorithm is nearly optimal.

As a corollary, we also obtain user-level private algorithms for the problem of learning discrete distributions. Specifically, we show that our results on mean estimation directly imply that one can learn a discrete distribution up to a total variation distance of $\tilde{O}\left(\sqrt{\frac{d}{mn} + \frac{d^2}{mn^2 \eps^2}}\right)$ as long as $n \ge O(\frac{1}{\eps} \log \frac{1}{\delta})$ (see Corollary \ref{cor:user_main} for details). 
Hence, our work extends the algorithm of Liu et al.~\cite{liu2020discrete}, which provided comparable bounds for user-level private distribution learning but only when $n \ge \tilde{O}(\frac{1}{\eps} \sqrt{d \log \frac{1}{\delta}})$.
Our mean estimation algorithm is also directly applicable to a variant of stochastic gradient descent, and to the problems of empirical risk minimization and stochastic convex optimization. Indeed, Levy et al.~\cite{levy2021user} uses mean estimation as a black-box to solve these problems with user-level privacy. As this part is essentially unchanged from Levy et al.~\cite{levy2021user}, we briefly discuss how we extend these results in Appendix \ref{sec:erm_and_stuff}.

Finally, we prove that our private algorithm for learning discrete distributions is nearly optimal in sample complexity when $\frac{1}{\delta}$ is approximately polynomial in the dimension, both when the number of users $n$ is low or high. This resolves an open question of Liu et al.~\cite{liu2020discrete}, who provided a comparable lower bound for $(\eps, 0)$-differential privacy and asked whether one could extend the lower bound to $(\eps, \delta)$-differential privacy for nonzero $\delta$.

We reiterate that all of algorithms work under a certain notion of robustness as well. Namely, we show that if the samples for each individual user are slightly correlated, and even if $49\%$ of the users have all of their data adversarially corrupted, we can still estimate the mean of the distribution up to near-optimal accuracy in the regime when the number of users $n$ is at most $O(\sqrt{d})$.

\subsection{Related Work}
Differential privacy was first introduced by Dwork et al.~\cite{dwork2006calibrating} and since then has drawn significant attention. There have been several previous works on this topic, with a main focus on item-level differential privacy~\cite{dwork2014algorithmic,wasserman2010statistical,hay2009boosting,dwork2008differential,abadi2016deep,mcsherry2007mechanism,friedman2010data,dwork2010boosting,xiao2010differential,dwork2009differential}. Some of this work has focused on high-dimensional private mean estimation~\cite{kamath2019highdimensional, liu2021robust, hopkins2022robust}, with the latter two even allowing for adversarial robustness. However, these works focus on item-level privacy, and assume the number of data points is at least linear in the dimension.

Recently, Liu et al.~\cite{liu2020discrete} studied the problem of learning discrete distributions with user-level differential privacy and provided the first non-trivial user-level differentially private algorithm for this problem. Note that private learning of discrete distributions can be done via a private mean estimation algorithm. 
Later, Levy et al.~\cite{levy2021user} extended the work of Liu et al.~\cite{liu2020discrete} and provided user-level differential privacy mechanisms for several learning tasks, namely, high-dimensional mean estimation, empirical risk minimization with smooth losses, stochastic convex optimization, and learning hypothesis class with finite metric entropy. A core part of these algorithms was their algorithm for user-level differentially private mean estimation.

In the item-level setting, Huang et al.~\cite{huang2021itemlevel} studied mean estimation and provided guarantees similar to those of Levy et al.~\cite{levy2021user}: both papers require the number of samples (in the item-level setting) to be at least $\sqrt{d}/\eps$.

User-level differential privacy has also been recently studied for learning models in the federated learning setting~\cite{augenstein2019generative,mcmahan2018general,mcmahan2017learning}. The concept of bounding user contribution motivated by user-level differential privacy has also been studied in the contexts of histogram estimation~\cite{amin2019biasvarianceuser} and SQL~\cite{wilson2019differentially}. 

\paragraph{Independent Work:} We note that independent work by Ghazi et al.~\cite{ghazi2021independentwork} also studied the problem of user-level learning, and proved that one can solve mean estimation as well as a very broad class of learning problems using $\tilde{O}(\frac{1}{\eps} \log \frac{1}{\delta})$ users and polynomially many samples per user. However, unlike our work they do not obtain tight user-sample tradeoffs, and do not provide robust guarantees.
Another independent work by Tsfadia et al.~\cite{tsfadia2022independentwork} provided a general framework that can recover similar guarantees to ours in the item-level setting, even with $O(\eps^{-1} \log \delta^{-1})$ items, and can also provide robust guarantees similar to ours. However, their algorithm is slower and requires time quadratic in the number of users $n$, whereas our runtime is almost linear in $n$.
We remark that the techniques of both \cite{ghazi2021independentwork} and \cite{tsfadia2022independentwork} are significantly different from ours.

\section{Our Techniques} \label{subsec:techniques}

In this section we explain our algorithmic techniques and the challenges we face. Our main innovation is developing a new private estimation method when the number of users is very small, in the case where each user only has $1$ item. From here, we show that it is simple to convert this to a general user-level private algorithm by having each user take the mean of their samples and treating the mean as a single item. So, for the upper bounds, we focus on item-level differential privacy.

\paragraph{Learning a ball with $O(\frac{1}{\eps} \log \frac{1}{\delta})$ points:} We create an item-level private algorithm that, if most items are concentrated in an unknown ball of radius $1$, estimates the ball up to radius $\sqrt{d}$. While this may seem like a very poor estimate, this is in fact nearly optimal when the number of items is only $O(\frac{1}{\eps} \log \frac{1}{\delta})$. In addition, in the user-level setting, the sample mean of each (uncorrupted) user's data will be tightly concentrated if the number of samples per user is large, so if we can get concentration $\gamma/\sqrt{d}$ for the mean of each user, then we can privately estimate the distribution mean up to error $\gamma$.

The main challenge here is that previous algorithms decomposed the distribution into $1$-dimensional components and combined the estimate in each dimension. By the strong composition theorem (see Appendix \ref{sec:composition} or \cite{dwork2014algorithmic} for the statement), one needs approximately $O(\eps/\sqrt{d})$-level differential privacy in each dimension to ensure $\eps$-level differential privacy overall, and it is well-known that this implies at least $\sqrt{d}/\eps$ samples are needed. Hence, we need a method that avoids decomposing the distribution into $1$-dimensional pieces.

We develop an algorithm loosely based on the exponential mechanism. For simplicity, we assume all coordinates of the points are bounded in $[-R, R]$ for a very large $R$, and we are okay with dependencies on $\log R$. The rough outline of the standard exponential mechanism, used for $1$-dimensional mean estimation, is to split the region $[-R, R]$ into intervals of length $1$, and then select an interval $[a, a+1]$ with probability proportional to $e^{\eps \cdot f(a)}$, where $f(a)$ is the number of data points in the interval $[a, a+1]$. This method will only need $O(\log R/\eps)$ samples.

In $d$ dimensions, a first attempt at modifying the $1$-dimensional algorithm is as follows. Sample each point $p \in [-R, R]^d$ with probability proportional to $e^{\eps \cdot f(p)}$, where $f(p)$ is the number of data points that are within some distance $T$ of $p$, for some choice of $T$. Unfortunately, this attempt runs into a major problem. 
Namely, even for points not close to any of the data points, we still have that $e^{\eps \cdot f(p)} = e^0 = 1$, and there is a roughly $R^d$ volume of such points. One can show that unless we have roughly $O(\log (R^d)/\eps) = O(d \log R/\eps)$ samples, we will almost never sample from any point with large $f(p)$. However, this requires a linear sample complexity in the dimension $d$, whereas we do not want any dependence on $d$.

A first attempt at fixing this is to only sample points $p$ if $f(p) \ge 1$. However, we lose differential privacy, since changing a single data point may convert $f(p)$ from $1$ to $0$ for many points $p$, which means that the probability of sampling $p$ goes from proportional to $e^{\eps}$ to proportional to $0$, which is not an $e^{\pm \eps}$ multiplicative change. To fix this, we add a ``garbage bucket'' that we sample from with a certain probability. By choosing the garbage bucket probability correctly, we ensure approximate $(\eps, \delta)$-differential privacy, because we show that we sample from the garbage bucket with probability significantly more than sampling $p$ with $f(p) = 1$. In addition, we show that as long as $T = O(\sqrt{d})$, we sample from a point with large $f(p)$ with high probability and select the garbage bucket with very low probability, assuming that most points are close together. To establish this, we prove a technical lemma that the intersection of several close balls of radius $\sqrt{d}$ still has a comparatively large volume, which means there is a large volume of points $p$ with large $f(p)$. 

Unfortunately, this sampling method is very hard to run, and in fact takes at least exponential time, because the way that $n$ balls intersect in high dimensions can be very complicated. To make this efficient, we develop an efficient rejection sampling procedure. While this rejection sampling procedure is fast if all data points are close together, it can be very slow for worst-case datasets, so we have to stop the rejection sampling after a certain number of rounds. Unfortunately, stopping the rejection sampling after a fixed number of rounds can cause privacy to break, because changing a single point may cause the probability of acceptance from the rejection sampling algorithm to significantly change. To fix this, we develop a privacy-preserving rejection sampling algorithm, by initializing a geometric random variable $X$ and running the rejection sampling algorithm $X$ times. By setting the decay rate of $X$ appropriately, we are able to successfully obtain privacy while ensuring an efficient runtime.

The overall procedure is in Algorithm \ref{alg:dp_estimate_1}.
To summarize a simplified version of the algorithm, we pick a uniformly random point $x_i$ and a uniformly random point $p$ within $\sqrt{d}$ of $x_i$. We accept this point with probability roughly $\frac{n}{f(p)} \cdot e^{\eps(f(p)-n)}$. This will end up being equivalent to picking a point $p$ proportional to $e^{\eps \cdot f(p)}$. With some probability roughly $\frac{1}{\delta \cdot e^n}$ we pick a garbage bucket instead of a point $p$, which will be like picking the garbage bucket with probability proportional to $\frac{1}{\delta}$. If we picked a point $p$ but did not accept it, we repeat the algorithm $X$ times, where $X$ is geometric random variable with mean roughly $e^{\sqrt{\log n}}$, until we either pick an accept a point $p$ or pick a garbage bucket. If this procedure fails after $X$ times, we give up and pick the garbage bucket.

\paragraph{Interpolation:} Our techniques for establishing a tight user-sample tradeoff in the intermediate regime of $\frac{1}{\eps} \log \frac{1}{\delta} \ll n \ll \sqrt{d} \cdot \frac{1}{\eps} \log \frac{1}{\delta}$ users are based on combining our methods when $n = O(\frac{1}{\eps} \log \frac{1}{\delta})$ with the Fast Johnson-Lindenstrauss method~\cite{fastjl} and strong composition theorems. Since this idea of using Fast Johnson-Lindenstrauss was used in Levy et al.~\cite{levy2021user} (in the regime $n \gg \sqrt{d} \cdot \frac{1}{\eps} \log \frac{1}{\delta}$) and we apply it in a similar manner, we defer our discussion of this part to Subsection \ref{subsec:interpolation} and the appendix.

\paragraph{Lower Bound Techniques:}
Our main contribution is a tight lower bound for user-level private learning of an unknown discrete distribution, which can be thought of as a special case of mean estimation. Hence, we also obtain a tight lower bound for user-level mean estimation.

For mean estimation, a user-level lower bound follows based on the item-level lower bound techniques of Kamath et al.~\cite{kamath2019highdimensional}. However, their techniques importantly use the fact that each coordinate is independent (such as spherical Gaussian distributions), whereas in the setting of learning discrete distributions, the coordinates are in fact negatively correlated. However, we are able to combine their techniques with a Poissonization method. Namely, rather than treating each user as having received $m$ samples from a distribution over $[d]$, we assume that each user has received $\text{Pois}(m)$ samples. In this case, if the distribution were known to have $p_i$ as the probability of sampling $i$, then the number of samples that equal $i$ would have distribution $\text{Pois}(m \cdot p_i)$ and would be independent across $i$. In the case when each $p_i$ is generated independently and each user is given a independent sample from $\text{Pois}(m \cdot p_i)$ for all $i \in [d]$, we are able to apply the methods of Kamath et al.~\cite{kamath2019highdimensional}, albeit with significantly different computations.

Our final piece is to show that if we can privately learn discrete distributions, then we can privately learn the Poisson variables as well, which allows us to also get a lower bound for discrete distribution learning. For this, we treat each $p_i$ as an independent random variable, and consider the discrete distribution $\tilde{p}$ with probability $\tilde{p}_i := p_i/(p_1+\cdots+p_d),$ i.e., normalized to have overall probability $1$. From here, we show how to simulate the multivariate Poisson distribution $\text{Pois}(m \cdot p_i)_{i = 1}^{d}$ given $m$ samples from the discrete distribution $\tilde{p}$, which we use to prove that privately learning discrete distributions is at least as hard as privately learning Poisson distributions. One issue is that we do not know the sum $\sum_{i = 1}^{d} p_i$ since the sum has been normalized thorugh $\tilde{p}_i$, but we show how to deal with this by privately learning a one-dimensional Poisson.

\section{Item-level Private Algorithm} \label{sec:item}
In this section, we assume that every user only has a single point (item) in $\BR^d$, and we develop a robust and differentially private algorithm for estimating a crude approximation of their location when the number of points is very small. In the next section, we use the methods of the item-level algorithm to develop a user-level differentially private algorithm.

In subsection \ref{subsec:exponential}, we give an algorithm for estimating $O(\frac{1}{\eps} \log \frac{1}{\delta})$ points that may run in exponential time. In subsection \ref{subsec:polynomial}, we improve the algorithm to run efficiently. Finally, in subsection \ref{subsec:interpolation}, we give an algorithm for estimating $n$ points accurately, when $n$ is in between $O(\frac{1}{\eps} \log \frac{1}{\delta})$ and $O(\sqrt{d} \cdot \frac{1}{\eps} \log \frac{1}{\delta})$. We always assume that $\eps, \delta \le \frac{1}{3}$.

\subsection{Computationally Inefficient Algorithm} \label{subsec:exponential}

We start by stating the theorem, which establishes an $(\eps, \delta)$-differentially private estimation algorithm, but with no guarantees on the runtime.

\begin{theorem} \label{thm:exponential_main}
    Let $n \ge C \cdot \frac{1}{\eps} \log \frac{1}{\delta}$ for a sufficiently large constant $C$, and fix a positive real number $r$. Then, there exists an algorithm over $n$ points $x_1, \dots, x_n \in \BR^d$ with the following guarantees.
    \begin{enumerate}
        \item The algorithm is $(\eps, \delta)$-differentially private over arbitrary datasets $\{x_1, \dots, x_n\}$.
        \item If there exists a ball of radius $r$, centered at some unknown point $x$, that contains at least $\frac{2}{3}$ of the points $x_1, \dots, x_n$, then with success probability $1$ the algorithm will return a point $x'$ such that $\|x'-x\|_2 \le O(r \cdot \sqrt{d})$.
    \end{enumerate}
\end{theorem}

The algorithm works as follows. 
Define $B_i$ to be the (closed) ball of radius $r\sqrt{d}$ around each point $x_i$, and for each point $p \in \mathbb{R}^d,$ define $f(p)$ to be the number of balls $B_i$ containing $p$. Now, consider the following (unnormalized) distribution over $\BR^d$. Let the distribution have density proportional to $\mathcal{D}(p) = e^{\eps \cdot \min(f(p), 2n/3)}$ if $f(p) > 0$ and density $\mathcal{D}(p) = 0$ if $f(p) = 0.$ Finally, we add a ``garbage bucket'' with a point mass proportional to $\frac{4}{\delta} \cdot V_B$ for $V_B$ the volume of each ball $B_i$.

To establish accuracy of our algorithm, we will require the following lemma, which states that the volume of the intersection of the balls $B_i$ is large if they are all close together. This will become crucial verifying the second guarantee of Theorem \ref{thm:exponential_main}, because it establishes that the probability of picking a point in the intersection of the balls that are close together will overwhelm the probability of picking a point far away from these balls.

\begin{lemma} \label{lem:volume}
    Suppose there are $n$ points $x_1, \dots, x_n$ each of distance at most $r$ from some point $x$, and let $B_i$ be the ball of radius $r\sqrt{d}$ around $x_i$. Then, $\text{Vol}(\bigcap_{i} B_i) \ge e^{-10\sqrt{\log n}} \cdot V_B$, where $V_B$ is the volume of each ball $B_i$.
\end{lemma}

In addition, we will need to check privacy (i.e., Guarantee 1 of Theorem \ref{thm:exponential_main}). Roughly, this will follow because the garbage bucket will have volume much more than the volume of points with $f(p) = 1$ (which are the only points which change by more than an $e^{\pm \eps}$ factor after a single item is modified). However, the volume of the garbage bucket will be much less than the volume of points with large $f(p)$ after scaling by the $e^{\eps \cdot \min(f(p), 2n/3)}$ factor by Lemma \ref{lem:volume}. This means that if in fact $\frac{2}{3}$ of the points are concentrated in a ball of radius $r$, we select the garbage bucket with very low probability, so our algorithm will be accurate. Combined, these observations will allow us to prove Theorem \ref{thm:exponential_main}.

The full proofs of Lemma \ref{lem:volume} and Theorem \ref{thm:exponential_main} are deferred to Appendix \ref{subsec:exponential_proofs}.

\subsection{Polynomial-time Algorithm} \label{subsec:polynomial}

The main issue with the algorithm described in Subsection \ref{subsec:exponential} is that it is difficult to run efficiently. In this subsection, we demonstrate a modification of the previous algorithm that ensures the algorithm runs efficiently, even for worst-case datasets.

We now describe our faster algorithm. First, initialize a random variable $X \sim Geom(1/N)$, where $N = \Theta(e^{10 \sqrt{\log n}} \cdot \alpha^{-1})$. Here, $n$ is the number of data points, and $\alpha$ will represent our (preliminary) failure probability. Now, we consider the following rejection sampling-based algorithm. We first pick an index $i \in [n+1]$, where we choose each $i \in [n]$ with probability proportional to $e^{\eps \cdot 2n/3}$, and choose $n+1$ with probability proportional to $\frac{4 n}{\delta}$. If we picked $i \le n$, we sample a uniform point $p$ in $B_i$, and accept $p$ with probability $\frac{1}{3} \cdot \frac{n}{f(p)} \cdot e^{\eps (\min(f(p), 2n/3)-2n/3)}$. Else, if we picked $n+1$, we pick a garbage bucket that we label $G_1$ and accept with probability $\frac{1}{3}$. If we reject, then we retry the algorithm. We retry this algorithm a total of $X \sim Geom(1/N)$ times, and if the algorithm rejects all $X$ times, we automatically select a second garbage bucket $G_2$ as the estimate. We provide the pseudocode for this procedure in Algorithm \ref{alg:dp_estimate_1}.

\begin{figure}
\centering
\begin{algorithm}[H]
    \caption{: \textbf{\textsc{DP-Estimate-1}}($x_1, \dots, x_n, \alpha, \eps, \delta$).}
    \label{alg:dp_estimate_1}
    \begin{algorithmic}[1] 
            \State Let $N \leftarrow \Theta(e^{10 \sqrt{\log n}} \cdot \alpha^{-1})$, and sample $X \sim Geom(1/N)$.
            \For{$\text{rep} = 1$ to $X$}
                \State Sample $i \sim [n+1]$, with each $1 \le i \le n$ proportional to $e^{\eps \cdot 2n/3}$ and $n+1$ proportional to $\frac{4n}{\delta}$.
                \If{$i \le n$}
                    \State Sample $p \sim \text{Unif}[B_i]$, where $B_i$ is the ball of radius $r \cdot \sqrt{d}$ around $x_i$.
                    \State Let $f(p) = \#\{j \in [n]: \|p-x_j\|_2 \le r \sqrt{d}\}$.
                    \State \textbf{Return} $p$ with probability $\frac{1}{3} \cdot \frac{n}{f(p)} \cdot e^{\eps (\min(f(p), 2n/3)-2n/3)}$.
                \Else
                    \State \textbf{Return} $G_1$ with probability $\frac{1}{3}$.
                \EndIf
            \EndFor
            \State \textbf{Return} $G_2$.
    \end{algorithmic}
\end{algorithm}
\caption{Differentially private estimation algorithm taking $n = O(\frac{1}{\eps} \log \frac{1}{\delta})$ points $x_1, \dots, x_n \in \BR^d$. If there exists $x$ within $r$ of at least $2/3$ of the points $\{x_i\}$, then the algorithm outputs a point within $O(r\sqrt{d})$ of $x$, with probability $1-\alpha$.}
\label{fig:dp_estimate_1}
\end{figure}

To verify the accuracy and privacy of our algorithm, we must verify the following guarantees.
\begin{enumerate}
    \item Conditioned on the rejection sampling algorithm not returning $G_2$, the distribution of the output is the same as our computationally unbounded algorithm.
    \item If at least $\frac{2}{3}$ of the points are in a ball of radius $r$, the algorithm returns $G_2$ with very low probability.
    \item If we change a dataset to a neighboring dataset, the probability that the algorithm outputs $G_2$ is not affected significantly.
\end{enumerate}

The first and second guarantees together ensure accuracy of the algorithm. This is because if $\frac{2}{3}$ of the points are in a ball of radius $r$, then the algorithm is very unlikely to output $G_2$, and otherwise, the output distribution is approximately the same as in the inefficient algorithm, which we know is accurate.

The third guarantee ensures privacy of the algorithm's output. This is because the algorithm's output is equivalent to first determining whether to output $G_2$ (which the third guarantee roughly ensures is private), and then if not running the inefficient algorithm, which is also private.

Finally, we runtime will be very fast, as the number of iterations of the for loop is roughly $O(e^{O(\sqrt{\log n})} \cdot \alpha^{-1})$ and the actual for loop can be implemented in $O(nd)$ time.

From here, we are able to establish our main result for item-level privacy with few users. We defer the proof of the theorem below and our three intermediate guarantees above to Appendix \ref{subsec:polynomial_proofs}.

\begin{theorem} \label{thm:main}
    Suppose that we are given input points $x_1, \dots, x_n \in \BR^d$, where $n \ge C \cdot \frac{1}{\eps} \cdot \log \frac{1}{\delta}$ for a sufficiently large constant $C$ and for $\eps, \delta, \alpha \le \frac{1}{3}$. Then, if at least $2/3$ of the points $x_i$ are contained in a ball centered around some $x$ of radius $r$, where $r$ is known but $x$ is unknown, Algorithm \ref{alg:dp_estimate_1} finds a point within $O(\sqrt{d} \cdot r)$ of $x$ with probability at least $1-O(\alpha)$. In addition, even for worst-case datasets, Algorithm \ref{alg:dp_estimate_1} runs in time $O(N \cdot n \cdot d) = O(n e^{\sqrt{10 \log n}} \cdot d \cdot \alpha^{-1})$ in expectation and is $(O(\eps), O(\delta))$-differentially private.
\end{theorem}

Note that if $\alpha^{-1} = (nd)^{o(1)}$, then our algorithm runs in time almost linear in $n$ and $d$, so the runtime is nearly optimal. In addition, as we will show in Lemma \ref{lem:lb_easy}, the choice of $n = C \cdot \frac{1}{\eps} \log \frac{1}{\delta}$ is optimal.

One potential downside of Algorithm \ref{alg:dp_estimate_1} is that the runtime depends inversely on the failure probability $\alpha$. If we wish to succeed with very low failure probability, the runtime can be very slow. This, however, can be fixed with only a slight loss in the number of samples. To do so, we let $1/3$ represent the ``initial'' failure probability. If we wish for an overall $1-\alpha$ failure probability, we run the algorithm $O(\log \frac{1}{\alpha})$ times, and return the coordinate-wise median of all returned points which are not $G_1$ or $G_2$. In this case we obtain the following corollary of Theorem \ref{thm:main}:

\begin{corollary} \label{cor:main}
    Consider the same setup as in Theorem \ref{thm:main}, except that $n \ge C (\frac{1}{\eps} \log \frac{1}{\delta} \log \frac{1}{\alpha})$. Then, for any $\alpha \le \frac{1}{3}$, there is an algorithm that runs in time $O(n^{1+o(1)} \cdot d \cdot \log \frac{1}{\alpha})$ and is $(\eps, \delta)$-differentially private even for arbitrary datasets, such that if at least $2/3$ of the points $x_1, \dots, x_n$ are in a ball of radius $O(r)$, with probability $1-\alpha$ the output is within $O(\sqrt{d} \cdot r)$ of this ball.
\end{corollary}

\begin{remark}
    As long as $\alpha \ge \exp\left(-d^{o(1)}\right),$ the required dependence of $n$ on $d$ is subpolynomial, i.e., $d^{o(1)}$, and the runtime is almost linear, i.e., $(nd)^{1+o(1)}$. Thus, we can even have cryptographically small failure probability in almost linear time.
\end{remark}

To obtain Corollary \ref{cor:main} from Theorem \ref{thm:main}, the privacy guarantee roughly follows from standard composition theorems of differential privacy. The accuracy guarantee follows from the fact that if at least $2/3$ of a set of points $p_1, \dots, p_k$ (where we set $k = O(\log \frac{1}{\alpha})$ and think of $p_i$'s as the returned points) are all within distance $K = O(\sqrt{d} \cdot r)$ of $x$, then the coordinate-wise median is also within $O(K)$ of $x$~\cite{1centerclustering}.

We defer all proofs, as well as pseudocode corresponding to Corollary \ref{cor:main}, to Appendix \ref{subsec:polynomial_proofs}.

\medskip

To summarize, we have devised an algorithm that can efficiently and privately approximate a ball of radius $r$ in the item-level setting, even if the number of points in the ball is tiny compared to the dimension and even if $1/3$ of the points are outliers. We reiterate that $1/3$ can be improved to any constant below $1/2$, and this will allow our final user-level private algorithm to be robust even if $49\%$ of users' data are corrupted. Although the approximation is quite poor, we will show that in the user-level setting, we can significantly improve the error if each user has enough samples.

We also remark that this problem is closely related to the private minimum $k$-enclosing ball problem, where the goal is to privately find a ball of some radius that approximately contains the maximum number of data points. This problem was first studied in developing the famous sample-and-aggregate framework in differential privacy~\cite{nissim2007smoothsensitivity}, and later in applications to private clustering~\cite{stemmer2018clustering, stemmer2020clustering, ghazi2020clustering}. Our algorithm provides an optimal $O(\sqrt{d})$-approximation to this problem in the item-level case even if there are only $O(\frac{1}{\eps} \log \frac{1}{\delta})$ points and $51\%$ of the points are in an unknown ball. 
We remark that all previous algorithms for this problem required at least $\Omega(\sqrt{d})$ points in the optimal ball: we provide the first algorithm with guarantees for $o(\sqrt{d})$ points.

\subsection{Interpolation} \label{subsec:interpolation}

In Subsections \ref{subsec:exponential} and \ref{subsec:polynomial}, we focused on the problem when $n \approx \frac{1}{\eps} \log \frac{1}{\delta}$, and obtained error $O(r \sqrt{d})$ in our estimate. 
In fact, this precisely extends the work of Levy et al.~\cite{levy2021user}, which roughly proved that if $n \ge \tilde{O}(\sqrt{d \log \frac{1}{\delta}}/\eps)$, we are able to estimate the sample mean up to error $r \cdot \tilde{O}\left(\frac{1}{\sqrt{n}} + \frac{\sqrt{d}}{n \cdot \eps}\right)$ (for simplicity, we ignore all logarithmic factors in the error, including $\log \frac{1}{\delta}$ factors).
To see why, if we set $n = O\left(\frac{1}{\eps} \log \frac{1}{\delta}\right)$, which is much smaller than $\sqrt{d \log \frac{1}{\delta}}/\eps$, we obtain the same error (up to logarithmic factors) as a function of $n$.
Hence, this elicits the question of whether we can extend~\cite{levy2021user} to the full regime of $n \ge O(\frac{1}{\eps} \log \frac{1}{\delta})$.
In this subsection, we successfully do so, by combining our algorithm from Subsection \ref{subsec:polynomial} with a method modified from~\cite{levy2021user}.

The idea of Levy et al.~\cite{levy2021user} was to use the ``Fast Johnson-Lindenstrauss'' technique~\cite{fastjl}, which randomly rotates the data efficiently, and then performs a private $1$-dimensional mean estimation algorithm in each coordinate.
(We defer the details of this random rotation procedure to Appendix \ref{subsec:interpolation_proofs}.)
The main advantage of this procedure is that if the points are concentrated in a ball of radius $r$, with probability at least $1-\alpha$ after the random rotation, the points are concentrated in an interval of length $O(r \cdot \sqrt{\log (d \cdot n/\alpha)/d})$ in all $d$ dimensions. Hence, Levy et al.~\cite{levy2021user} reduces $d$-dimensional estimation to $1$-dimensional estimation, but where the radius in each dimension is only about $r/\sqrt{d}$.
In addition, this random rotation is known to be invertible, and is known to be computable in time $O(d \log d)$ for an individual vector, so it takes time $O(d \log d \cdot n)$ to perform on all of $x_1, \dots, x_n$.

Our algorithm proceeds in a similar manner to~\cite{levy2021user}, but we reduce to $k^2$-dimensional estimation for some $1 \le k^2 \le d$ and apply Theorem \ref{thm:main}, rather than to $1$-dimensional estimation.
First, we use Fast Johnson-Lindenstrauss to randomly rotate the data $x_1, \dots, x_n$.
Now, fix some $k \le \sqrt{d}$, and for each $1 \le s \le k^2,$ let $x_i^{(s)}$ be the point $x_i$ after we randomly rotate it and project it onto coordinates $(s-1) \cdot (d/k^2) + 1, \dots, s \cdot (d/k^2)$.
The rotated points will be concentrated in a interval of length roughly $r/\sqrt{d}$ in each dimension, so the points $x_i^{(s)}$ for each $s$ are concentrated in a ball of radius roughly $(r/\sqrt{d}) \cdot \sqrt{d/k^2} \approx r/k$. Hence, for each $s$, we can use Theorem \ref{thm:main} to estimate the points $x_i^{(s)}$ up to error $r/k \cdot \sqrt{d/k^2}$, because the points $x_i^{(s)}$ are in $d/k^2$ dimensions. As $s$ ranges from $1$ to $k^2$, the overall Euclidean error when we combine the dimensions will be $r/k \cdot \sqrt{d/k^2} \cdot \sqrt{k^2} = r\sqrt{d}/k$. Finally, since we are performing $k^2$ different applications of Theorem \ref{thm:main} (for each $1 \le s \le k^2$), the Strong Composition Theorem tells us we actually need $(\eps/(k\sqrt{\log \frac{1}{\delta}}), \delta/k^2)$-DP for each use so that the overall algorithm is $(\eps, \delta)$-DP. Hence, we will need roughly $\tilde{O}\left(\frac{k}{\eps} \cdot \log \frac{1}{\delta}\right)$ samples.


By setting the parameter $k$ properly, we obtain the following theorem, which we fully prove in Appendix \ref{subsec:interpolation_proofs}.

\begin{theorem} \label{thm:interpolation}
    For $O\left(\frac{1}{\eps} \log \frac{1}{\delta}\right) \le n \le O\left(\sqrt{d} \cdot \frac{1}{\eps} \log \frac{1}{\delta}\right)$, there is an $(O(\eps), O(\delta))$-differentially private algorithm on $x_1, \dots, x_n$ in the worst case, such that if at least $2/3$ of the points $x_1, \dots, x_n$ are all in a ball of radius $r$ centered at some unknown point $x$, then with probability $1-\alpha$, the algorithm will return a point within 
    $$O\left(\frac{r \cdot \sqrt{d} \cdot \frac{1}{\eps} \cdot \log \frac{d}{\delta} \sqrt{\log \frac{1}{\delta}} \cdot \sqrt{\log \frac{dn}{\alpha}}}{n}\right)$$
    of $x$. Moreover, the runtime of the algorithm is $O(n^{1+o(1)}d \cdot \alpha^{-1} + nd \log d)$.
\end{theorem}

\begin{remark}
    If one wishes for a runtime that is almost linear in $n$ and has logarithmic dependence on the failure probability $\alpha$, as in Corollary \ref{cor:main}, this can be done by dividing the privacy parameters by $\log \frac{1}{\alpha}$. So, if the number of samples $n$ is at least $O\left(\frac{1}{\eps} \log \frac{1}{\delta} \log \frac{1}{\alpha}\right),$ the error is still $\tilde{O}\left(\frac{r\sqrt{d}}{n} \cdot \frac{1}{\eps} \left(\log \frac{1}{\delta}\right)^{3/2}\right)$ with probability $1-\alpha$, where the $\tilde{O}$ hides logarithmic factors in $d$, $\alpha^{-1}$, and $n$, and the runtime is improved to $O(n^{1+o(1)} d \cdot \log \frac{1}{\alpha} + nd \log d)$.
\end{remark}

Theorem \ref{thm:interpolation} will be crucial in establishing tight bounds in the user-level case, when the number of users is between $O(\frac{1}{\eps} \log \frac{1}{\delta})$ and $\sqrt{d}$. We also note that, as in Subsection \ref{subsec:polynomial}, this also applies to (item-level) privately approximating minimum $k$-enclosing ball, providing a $\tilde{O}(\frac{\sqrt{d}}{\eps \cdot n})$ multiplicative approximation when the number of points $n$ is between $O(\frac{1}{\eps} \log \frac{1}{\delta})$ and $\sqrt{d}$, which turns out to be optimal.

\section{User-Level Private Algorithm} \label{sec:user}

In the user-level setup, we suppose that there are $n$ users, each of which is given $m$ data samples from $\BR^d$. Let $X_{i, j} \in \BR^d$ represent the $j$th sample of the $i$th user. We suppose that each sample $X_{i, j}$ has marginal distribution $\mathcal{D}$, which has unknown mean $\mu \in \BR^d$, but with a promise that $\BE\left[\|X_{i, j}-\mu\|_2^2\right] \le r^2$ for some known $r > 0$.
In addition, we assume that the data is independent across users.
However, among the $m$ samples of any individual uncorrupted user, the data may not be entirely independent, but the samples do not have significant covariance between them. Namely, for any user $i$ and samples $j \neq j'$, $\BE\left[\langle X_{i, j}-\mu, X_{i, j'}-\mu \rangle\right] \le \frac{1}{m} \cdot r^2.$ We do not make any restrictions on how negative $\BE\left[\langle X_{i, j}-\mu, X_{i, j'}-\mu \rangle\right]$ could be.
Finally, we allow for up to $n/4$ of the users to have their data corrupted adversarially, which means that after the data $X = \{X_{i, j}\}$ has been generated, an adversary may look at $X$, and then choose $n/4$ of the users and alter all of their data in any way.
We remark that $n/4$ could be increased to $\kappa \cdot n$ for any constant $\kappa < 1/2,$ with minor modifications to the overall algorithm.


Our results in the user-level setting are a simple extension of the item-level algorithms (either Theorem \ref{thm:main}, Corollary \ref{cor:main}, or Theorem \ref{thm:interpolation}).
Indeed, the algorithm simply takes the mean of each user's data, and then applies the item-level algorithm on the user means.
So, we will just directly state the main results.
For the setup as described in the above paragraph, we prove the following.

\begin{theorem} \label{thm:user_main}
    Suppose that $C \cdot \frac{1}{\eps} \log \frac{1}{\alpha \cdot \delta} \le n \le \sqrt{d} \cdot \frac{1}{\eps} \log \frac{1}{\delta}.$
    Then, there is an algorithm on a dataset $\{X_{i, j}\}$, where $1 \le i \le n,$ $1 \le j \le m$, and each $X_{i, j} \in \BR^d$, with the following properties:
\begin{enumerate}
    \item If the data is generated via the procedure described above, then the algorithm reports a point $\hat{\mu}$ such that with probability at least $1-\alpha$, 
\begin{equation*}
    \|\hat{\mu}-\mu\|_2 \le \min\left(O\left(r \cdot \sqrt{\frac{d}{m}}\right), O\Biggr(r \cdot \frac{\sqrt{d}\cdot \log \frac{d}{\delta} \sqrt{\log \frac{1}{\delta}} \cdot \sqrt{\log \frac{dn}{\alpha}}}{\eps \cdot n \cdot \sqrt{m}}\Biggr)\right).
\end{equation*}
    \item The algorithm is $(\eps, \delta)$ user-level differentially private, where the $i$th user is given $\{X_{i, j}\}_{j = 1}^{m}$, even if the datasets are not generated from the procedure above.
    \item The algorithm runs in expected time $O(mnd + nd \log d + n^{1+o(1)} d \alpha^{-1})$ over worst-case datasets.
\end{enumerate}
    In addition, the algorithm can be made to run in time $O(mnd + nd \log d + n^{1+o(1)} d \log \frac{1}{\alpha})$ and have failure probability $1-\alpha$, at the cost of an additional $\log \frac{1}{\alpha}$ factor in $n$.
\end{theorem}

More simply, by ignoring covariance between users, adversarial robustness, and all logarithmic factors, and by combining with the previous results of~\cite{liu2020discrete,levy2021user} we 
can state a simpler corollary of the main theorem. Due to a standard reduction, we additionally answer the question of learning a discrete distribution up to low total variation distance in the following corollary.

\begin{corollary} \label{cor:user_main}
    Let $n \ge C \cdot \frac{1}{\eps} \cdot \log \frac{1}{\alpha \cdot \delta}.$
    Suppose we have $n$ users, each with $m$ data samples $X = \{X_{i, j}\}$. Then, there is an algorithm taking as input $X$ that is $(\eps, \delta)$ user-level private in the worst case, and if each $X_{i, j}$ were drawn i.i.d. from an unknown distribution $\mathcal{D}$ entirely supported in an (unknown) ball of radius $r$, the algorithm learns the mean $\mu$ of $\mathcal{D}$ up to $\ell_2$ error
$$r \cdot \tilde{O}\left(\frac{\sqrt{d}}{\eps \cdot n \cdot \sqrt{m}} + \frac{1}{\sqrt{m \cdot n}}\right),$$
    with failure probability $\alpha$. Here, we use $\tilde{O}$ to drop all logarithmic dependencies on $\alpha^{-1}, \delta^{-1}, \eps^{-1}, m, n,$ and $d$.
    
    Finally, if instead each data sample $X_{i, j}$ were drawn from a discrete set $[d]$, there is an $(\eps, \delta)$ user-level private algorithm such that, if each sample were drawn i.i.d. from a distribution $\mathcal{D}$ over $[d]$, we could, with failure probability $\alpha$, learn $\mathcal{D}$ up to total variation distance
$$\tilde{O}\left(\frac{d}{\eps \cdot n \cdot \sqrt{m}} + \sqrt{\frac{d}{m \cdot n}}\right).$$
\end{corollary}

Corollary \ref{cor:user_main} provides an optimal (due to the lower bounds in Section \ref{sec:lb}) user-level private algorithm for mean estimation and discrete distribution learning, regardless of the number of users $n$ or number of samples $m$ per user.
Note that Corollary \ref{cor:user_main} holds both in the low- and high-user setting.

We prove both Theorem \ref{thm:user_main} and Corollary \ref{cor:user_main}, and also provide pseudocode,
in Appendix \ref{subsec:user_proofs}.

\section{Lower Bound} \label{sec:lb}

In this section, we provide lower bounds complementing our results from the previous sections. The first lower bound is simple and states that any private learning of the distribution requires at least $\Omega(\frac{1}{\eps} \log \frac{1}{\delta})$ users under reasonable conditions of $\eps, \delta$.

\begin{lemma} \label{lem:lb_easy}
    Suppose we have $n$ users, each with $m$ data points all contained in an unknown ball of radius $1$. Then, learning this ball up to any finite error $T$, with probability $2/3$ and with $(\eps, \delta)$ user-level privacy, cannot be done if $n \le \frac{1}{3} \cdot \frac{1}{\eps} \log \frac{1}{\delta}$, assuming that $\delta \le \eps^2$.
\end{lemma}

Next, we provide a lower bound for privately learning discrete distributions in the user-level setting, answering an open question of Liu et al.~\cite{liu2020discrete} by allowing $\delta$ to be just moderately small as opposed to $0$. Our overall method is similar to the lower bounds of Kamath et al~\cite{kamath2019highdimensional}, but with an additional ``Poissonization trick'' to overcome the issue of lack of independence between coordinates. In addition, we also must provide a reduction from learning discrete distributions to learning Poisson variables.

\begin{theorem} \label{thm:lb_main}
    Suppose we have $n$ users and $m$ samples per user, where each sample $X_{i, j} \in [d]$ represents the $j$th sample of user $i$.
    Suppose there exists an $(\eps, \delta)$ user-level differentially private algorithm $M$ that takes as input $X = \{X_{i, j}\}$, where $i$ ranges from $1$ to $n$ and $j$ ranges from $1$ to $m$, and where $\delta = o\left(\frac{\eps}{d^{3/2} \cdot m \cdot n^{1/2} \cdot \log^2 (\eps^{-1} \cdot mnd)}\right)$. Then, there exists a distribution $\tilde{p}$ over $[d]$, such that if $X_{i, j} \overset{i.i.d.}{\sim} \tilde{p}$, then with probability at least $1/3$, the total variation distance between $M(X)$ and $\tilde{p}$ is at least $$\Omega\left(\sqrt{\frac{d}{mn}} + \frac{d}{\eps \cdot n\sqrt{m} \cdot \log^7 (\eps^{-1} \cdot mnd)}\right).$$
\end{theorem}
    Note that Theorem \ref{thm:lb_main} is almost tight compared to Corollary \ref{cor:user_main}, up to logarithmic factors, and holds in both the low- and high-user setting. We did not attempt to optimize the logarithmic factors. Theorem \ref{thm:lb_main} also implies that the approximations in our item-level algorithms in Section \ref{sec:item} are also tight, as otherwise we could have used them to obtain better user-level private algorithms.

    We prove Lemma \ref{lem:lb_easy} and Theorem \ref{thm:lb_main} in Appendix \ref{sec:lb_proofs}.
    

\section*{Acknowledgments}

We thank Piotr Indyk and Ananda Theertha Suresh for helpful discussions. We also thank Srivats Narayanan for pointing us to the reference \cite{corsello2020impact}.

\newpage
\appendix
\section{Composition Theorems} \label{sec:composition}

In this section, we state the well-known weak and strong composition theorems of differential privacy for completeness.

First, we state the Weak Composition Theorem, which can be found, for instance, as Theorem 3.16 in Dwork and Roth~\cite{dwork2014algorithmic}.

\begin{theorem}{(Weak Composition Theorem)}
    Let $M_1, \dots, M_k$ be algorithms on a dataset $X$, such that each $M_i$ is $(\eps_i, \delta_i)$-differentially private for some $\eps_i, \delta_i \ge 0$. Then, the concatenation $M(X) = (M_1(X), \dots, M_k(X))$ is $(\eps, \delta)$-differentially private, where $\eps = \eps_1 + \cdots+\eps_k$ and $\delta=\delta_1+\cdots+\delta_k$.
\end{theorem}

Next, we state the Strong Composition Theorem, which can be found, for instance, as Theorem 3.20 in Dwork and Roth~\cite{dwork2014algorithmic}.

\begin{theorem}{(Strong Composition Theorem)}
    Let $\eps, \delta, \delta' \ge 0$ be arbitrary parameters. Suppose that $M_1, \dots, M_k$ be algorithms on a dataset $X$, such that each $M_i$ is $(\eps, \delta)$-differentially private. Then, the concatenation $M(X) = (M_1(X), \dots, M_k(X))$ is $\left(\sqrt{2k \log \frac{1}{\delta'}} \cdot \eps + k \eps \cdot (e^{\eps}-1), k \cdot \delta + \delta'\right)$-differentially private.
\end{theorem}

\section{Missing Proofs: Algorithms} \label{sec:alg_proofs}

In this section, we provide missing proofs for both the item-level and user-level algorithms. In Subsection \ref{subsec:exponential_proofs}, we prove all of the results from Subsection \ref{subsec:exponential}. In Subsection \ref{subsec:polynomial_proofs}, we prove all of the results from Subsection \ref{subsec:polynomial}. In Subsection \ref{subsec:interpolation_proofs}, we prove all of the results from Subsection \ref{subsec:interpolation}. Finally, in Subsection \ref{subsec:user_proofs}, we prove all of the results from Section \ref{sec:user}.

\subsection{Omitted Proofs from Subsection \ref{subsec:exponential}} \label{subsec:exponential_proofs}

In this subsection, our goal is to prove Theorem \ref{thm:exponential_main}. We start by proving Lemma \ref{lem:volume}, but before we do, we first prove an auxiliary result bounding the probability of a point on a ball having first coordinate too large. This result is relatively standard, but we prove it here to get explicit constants in our bound.

\begin{proposition} \label{prop:volume_ez}
    Let $n, d$ both be at least some sufficiently large constant, and choose a uniformly random point $(x_1, \dots, x_d)$ in the closed ball of radius $\sqrt{d}$, centered at the origin. Then, the probability that $x_1 \ge 5 \sqrt{\log n}$ is at most $\frac{1}{2n}$.
\end{proposition}

\begin{proof}
    It suffices to prove the same result with $(x_1, \dots, x_d)$ chosen from the sphere (i.e., where $\|x\|_2 = \sqrt{d}$) rather than the ball (i.e., where $\|x\|_2 \le \sqrt{d}$). It is well-known that $x_1$ has distribution $\frac{Z_1}{\sqrt{(Z_1^2+\cdots+Z_d^2)/d}},$ where $Z_1, \dots, Z_d$ are i.i.d. Standard Gaussians. So, the probability that $x_1 \ge 5 \sqrt{\log n}$ is at most the probability that $Z_1 \ge 2 \sqrt{\log n}$ plus the probability that $Z_1^2+\cdots+Z_d^2 \le \frac{d}{6},$ since one of these must occur in order for $x_1$ to be at least $5 \sqrt{\log n}$.
    
    The probability that a random Gaussian is at least $2 \sqrt{\log n}$, where $n$ is sufficiently large, is at most $e^{-(2 \sqrt{\log n})^2/2} \le \frac{1}{n^2}$. To bound the probability that $Z_1^2+\cdots+Z_d^2 \le \frac{d}{6}$, note that $Z_1^2+\cdots+Z_d^2 \sim \chi_d^2$, and it is well-known that the chi-square $\chi_d^2$ random variable has PDF
\[\frac{1}{2^{d/2} \Gamma(d/2)} \cdot x^{d/2-1} e^{-x/2} \le \frac{d^{O(1)} \cdot e^{d/2}}{d^{d/2}} \cdot x^{d/2-1} e^{-x/2},\]
    which for $x \le \frac{d}{6}$ is at most $d^{O(1)} \cdot (1/6)^{d/2} \cdot e^{d/2-d/12}$. Therefore, the probability of $Z_1^2+\cdots+Z_d^2 \le \frac{d}{6}$ is at most $d^{O(1)}\cdot (1/6)^{d/2} \cdot e^{d/2-d/12} \le e^{-0.4 \cdot d}$ for sufficiently large $d$.

    So, overall, the probability that $x_1 \ge 5 \sqrt{\log n}$ is at most $\frac{1}{n^2}+e^{-0.4 \cdot d}$, which is at most $\frac{2}{n^2} \le \frac{1}{2n}$ for sufficiently large $n$ unless $e^{-0.4 \cdot d} > \frac{1}{n^2}$, which means that $d \le 5 \log n$. But in this case, $x_1 \le \sqrt{d} \le \sqrt{5 \log n}$, so we have that $x_1 \ge 5 \sqrt{\log n}$ with probabability $0$. This concludes the proof.
\end{proof}

\begin{proof}[Proof of Lemma \ref{lem:volume}]
    First, by scaling we may assume WLOG that $r = 1$.

    Next, note that the ball of radius $\sqrt{d}-1$ around $x$ is contained in $\bigcap_i B_i$, and this ball has volume $V_B \cdot \left(\frac{\sqrt{d}-1}{\sqrt{d}}\right)^d \ge V_B \cdot e^{-5\sqrt{d}}$ for sufficiently large $d$. Therefore, if $2 \sqrt{\log n} \ge \sqrt{d}$, then $\text{Vol}(\bigcap_{i} B_i) \ge V_B \cdot e^{-5\sqrt{d}} \ge V_B \cdot e^{-10\sqrt{\log n}}$. So, we now assume that $2\sqrt{\log n} \le \sqrt{d}$.

    If a point $p$ is in $\bigcap_i B_i$, then $\|p - x_i\|_2^2 \le d$ for all $i$, which is equivalent to saying that $\|p\|_2^2 - 2 \langle p, x_i\rangle + \|x_i\|_2^2 \le d$. Therefore, since $\|x_i\|_2^2 \le 1$ for all $i$, if $\|p\|_2^2 \le d-1 + 2 \langle p, x_i \rangle$ for all $i$, then $p \in \bigcap_i B_i$. We will lowerbound the volume of points $p$ with $\|p\|_2^2 \le d-1 + 2 \langle p, x_i \rangle$ for all $i$.
    
    Suppose we pick a random point $p$ in the ball $\{\|p\|_2^2 \le d-1-10\sqrt{\log n}\}.$ Then, for a fixed $i$, if $\langle p, x_i \rangle \ge - 5 \sqrt{\log n}$, we have that $\|p\|_2^2 \le d-1+2 \langle p, x_i \rangle$. However, by rotational symmetry of the sphere, the probability of this event is equal to the probability that the first coordinate of $p$ is at least $-5 \sqrt{\log n}$. Since the ball has radius $\sqrt{d-1-10\sqrt{\log n}} \le \sqrt{d}$, the probability of this event is at least the probability that the first coordinate of a point picked on the $\sqrt{d}$-radius sphere is at least $-5 \sqrt{\log (n)}$. By Proposition \ref{prop:volume_ez}, the probability of this not occurring is at most $1/(2n),$ which means by the union bound, the probability that a random point $p$ with $\|p\|_2^2 \le d-1-10\sqrt{\log n}$ not being in all of the balls is at most $1/(2n) \cdot n = 1/2$, for any fixed points $x_1, \dots, x_n$. Therefore,
\[\text{Vol}\left(\bigcap_i B_i\right) \ge \frac{1}{2} \cdot V_B \cdot \left(\sqrt{\frac{d-1-10\sqrt{\log n}}{d}}\right)^d \ge \frac{1}{2} \cdot e^{-10 \sqrt{\log n}} \cdot V_B.\]
    To see why the last inequality holds, since $d$ is sufficiently large and $10\sqrt{\log n} \le 5 \sqrt{d}$, we have $\frac{1+10\sqrt{\log n}}{d} \le \frac{1}{2}$. Therefore, $\frac{d-1-10\sqrt{\log n}}{d} = 1-\frac{1+10\sqrt{\log n}}{d} \ge e^{-20 \sqrt{\log n}/d}$, so 
\[\left(\sqrt{\frac{d-1-10\sqrt{\log n}}{d}}\right)^d = \left(\frac{d-1-10\sqrt{\log n}}{d}\right)^{d/2} \ge e^{-20 \sqrt{\log n}/d \cdot d/2} = e^{-10 \sqrt{\log n}}.\] 
\end{proof}

\begin{proof}[Proof of Theorem \ref{thm:exponential_main}]
    First, we note that we can make the following three assumptions. 
    \begin{itemize}
        \item $\boldsymbol{r = 1}$. This is an assumption we may trivially make by scaling.
        \item $\boldsymbol{\eps \ge \delta}$. Note that if $\eps \le \delta,$ then any $(\delta/2, \delta/2)$-differentially private algorithm is known to be $(0, \delta)$-differentially private, and in this case we may even obtain an algorithm that succeeds using only $O(\frac{1}{\delta} \log \frac{1}{\delta}) \le O(\frac{1}{\eps} \log \frac{1}{\delta})$ samples.
        \item \textbf{In Condition 2, the success probability is only $\boldsymbol{1-\delta}$ instead of $\boldsymbol{1}$}.
        Let $\mathcal{A}$ be such an algorithm that succeeds with probability $1-\delta$.
        Since we are allowing for exponential time algorithms, if we are given black-box access to $\mathcal{A}$, we can modify algorithm to some $\mathcal{A}'$ as follows. First, $\mathcal{A}'$ verifies that at least $2/3$ of the points lie in a ball $B$ of radius $r=1$. If such a ball $B$ exists, the algorithm $\mathcal{A}'$ runs $\mathcal{A}$ conditioned on returning a point $x'$ of distance at most $O(\sqrt{d})$ from $B$, via rejection sampling. If $B$ does not exist, the algorithm $\mathcal{A}'$ will just run $\mathcal{A}$ once. This changes the total variation distance of the algorithm output's distribution by at most $\delta$ for any choice of $x_1, \dots, x_n$, so the algorithm is still $(\eps, O(\delta))$-differentially private.
    \end{itemize}

    Now, we state and prove two propositions, which establish accuracy and privacy, respectively, of our algorithm.

\begin{proposition} \label{prop:exponential_accurate}
    If $n \ge O\left(\frac{1}{\eps} \cdot \log \frac{1}{\alpha \cdot \eps \cdot \delta}\right),$ and at least $2/3$ of the points $x_1, \dots, x_n$ are all in some unknown ball $B$ centered at $x$ of radius $1$, then with probability at least $1-\alpha$, the algorithm chooses a point within $1+\sqrt{d}$ of $x$.
\end{proposition}

\begin{proof}
    Suppose we do not choose such a point. Then, either we chose the garbage bucket, or we chose a point $p$ that was not in $B_i$ for any $i$ with $\|x_i-x\|_2 \le 1$, in which case $f(p) \le \frac{1}{3} \cdot n$. Now, note that the volume of $\bigcap_{i: \|x_i-x\|_2 \le 1} B_i$ is at least $\frac{1}{2} \cdot e^{-10 \sqrt{\log 2n/3}} \cdot V_B \ge \frac{1}{2} \cdot e^{-10 \sqrt{\log n}} \cdot V_B$ by Lemma \ref{lem:volume} (by replacing $n$ with $2n/3$) and that $f(p) \ge \frac{2}{3} \cdot n$ for any $p \in \bigcap_{i: \|x_i-x\|_2 \le 1} B_i$. Therefore, the probability density function in this region is proportional to $e^{2/3 \cdot \eps \cdot n}$, whereas the mass of the garbage bucket is proportional to $\frac{4}{\delta}\cdot V_B$. So, we choose the garbage bucket with probability at most
\[\frac{\frac{4}{\delta}\cdot V_B}{e^{2/3 \cdot \eps \cdot n} \cdot \frac{1}{2} \cdot e^{-10 \sqrt{\log n}} \cdot V_B} = \frac{8}{\delta} \cdot e^{-(2/3 \cdot \eps \cdot n - 10 \sqrt{\log n})}.\]
    So, as long as $\frac{2}{3} \cdot \eps \cdot n - 10 \sqrt{\log n} \ge \log (8/\delta) + \log (2/\alpha) = \log (16/(\alpha \cdot \delta)),$ then we choose the garbage bucket with probability at most $\alpha/2$.
    
    Next, note that total volume of points for which $f(p) \neq 0$ is at most $n \cdot V_B$, which means that we choose a point $p$ not in $\bigcap_{i: \|x_i-x\|_2 \le 1} B_i$ with probability at most
\[\frac{e^{1/3 \cdot \eps \cdot n} \cdot n \cdot V_B}{e^{2/3 \cdot \eps \cdot n} \cdot \frac{1}{2} \cdot e^{-10\sqrt{\log n}} \cdot V_B} \le 2n \cdot e^{10\sqrt{\log n}} \cdot e^{-1/3 \cdot \eps \cdot n} = 2 e^{\log n + 10 \sqrt{\log n} - 1/3 \cdot \eps \cdot n} \le 2 e^{11 \log n - 1/3 \cdot \eps \cdot n}.\]
    So, as long as $\frac{1}{3} \cdot \eps \cdot n-11 \log n \ge \log(4/\alpha)$, then we choose a point $p$ not within $1+\sqrt{d}$ of $x$ with probability at most $\alpha/2$.
    
    Therefore, it suffices for $n$ to satisfy $\frac{2}{3} \cdot \eps \cdot n - 10 \sqrt{\log n} \ge \log \frac{16}{\alpha \cdot \delta}$ and $\frac{1}{3} \cdot \eps \cdot n-11\log n \ge \log \frac{4}{\alpha}$. For these to be true, it suffices that $n \ge \frac{30}{\eps} \cdot \sqrt{\log n}$, that $n \ge \frac{3}{\eps} \log \frac{16}{\alpha \cdot \delta}$, that $n \ge \frac{6}{\eps} \cdot \log \frac{4}{\alpha},$ and that $n \ge \frac{66}{\eps} \log n$. These are all clearly true as long as $n \ge C \cdot \frac{1}{\eps} \log \frac{1}{\alpha \cdot \eps \cdot \delta}$ for a sufficiently large constant $C$. This concludes the proof.
\end{proof}

\begin{proposition} \label{prop:exponential_dp}
    This algorithm is $(\eps, \delta)$-differentially private.
\end{proposition}

\begin{proof}
    We just need to see what happens when a single ball is moved. For any point $p$, the number of balls $B_i$ containing $p$ (which we called $f(p)$) does not change by more than $1$. Hence, the relative (unnormalized) density function at $p$ does not change by more than an $e^{\eps}$ factor, unless $f(p)$ changes from $0$ to $1$ or from $1$ to $0$. However, this only occurs for at most $V_B$ volume, so the unnormalized density of points $p$ for which $f(p)$ goes from $1$ to $0$ or from $0$ to $1$ is at most $e^\eps \cdot V_B$ (for each), so the total such unnormalized density is at most $2 e^{\eps} \cdot V_B \le 4 \cdot V_B$. So, the amount of normalized density cannot be more than $\delta$, since the garbage bucket has unnormalized density $\frac{4}{\delta} \cdot V_B$, which is at least $\frac{1}{\delta}$ times the total density of what we are modifying by more than an $e^{\eps}$ factor. Thus, we have $(\eps, \delta)$-differential privacy.
\end{proof}

By the initial WLOG assumptions, we may assume that $r = 1$, $\eps \ge \delta$ and that $\alpha = \delta$. So, by Proposition \ref{prop:exponential_accurate}, we only need $O(\frac{1}{\eps} \log \frac{1}{\delta})$ samples. This proves Theorem \ref{thm:exponential_main}.
\end{proof}

\subsection{Omitted Proofs and Pseudocode from Subsection \ref{subsec:polynomial}} \label{subsec:polynomial_proofs}

In this subsection, we state and prove all results omitted from Subsection \ref{subsec:polynomial}, and provide pseudocode for Corollary \ref{cor:main}.


First, we verify that Algorithm \ref{alg:dp_estimate_1} is in fact a valid algorithm. To do so, we have the following proposition.

\begin{proposition} \label{prop:valid_prob}
    For $n \ge C \cdot \frac{1}{\eps} \log \frac{1}{\eps}$, the value $\frac{1}{3} \cdot \frac{n}{f(p)} \cdot e^{\eps \cdot \left(\min(f(p), 2n/3)-2n/3\right)}$ is at most $\frac{1}{2}.$ Hence, it is a valid probability.
\end{proposition}

\begin{proof}
    If $f(p) \ge 2n/3$, this equals $\frac{n/3}{f(p)} \le \frac{1}{2}$, since $e^{\eps \cdot (\min(f(p), 2n/3)-2n/3)} = 1$. Otherwise, $$\frac{1}{3} \cdot \frac{n}{f(p)} \cdot e^{\eps \cdot (\min(f(p), 2n/3)-2n/3)} = \frac{1}{2} \cdot \frac{e^{\eps \cdot f(p)}/f(p)}{e^{\eps \cdot 2n/3}/(2n/3)},$$ so if we define $x = \eps \cdot 2n/3$ and $y=\eps \cdot f(p)$, then the probability equals $\frac{1}{2} \cdot \frac{e^y/y}{e^x/x}$. For $n \ge C \cdot \frac{1}{\eps} \log \frac{1}{\eps},$ $e^x/x \ge \eps^{-2}$. However, if $y \le 1$ then $e^y/y \le e/y = O(\frac{1}{\eps})$ since $y = \eps \cdot f(p) \ge \eps$, and if $1 \le y \le x,$ then $e^t/t$ is increasing for $t \ge 1$, so $e^y/y \le e^x/x$. Therefore, we always have that $\frac{e^y/y}{e^x/x} \le 1$, which concludes the proof. 
\end{proof}

Thus, our algorithm is actually a valid algorithm, and in fact in each round of the rejection sampling, we never accept with probability more than $\frac{1}{2}.$

Next, we verify the three guarantees listed in Subsection \ref{subsec:polynomial}.
We start by showing that Algorithm \ref{alg:dp_estimate_1} returns a point from the same distribution as the previous algorithm of Subsection \ref{subsec:exponential}, conditioned on not choosing $G_2$. Specifically, we show the following proposition.

\begin{proposition} \label{prop:no_G2_same_alg}
    Consider a single rejection round of Algorithm \ref{alg:dp_estimate_1}. Then, conditioned on accepting $p$ or returning $G_1$, we sample each point $p$ with density proportional to $e^{\eps \cdot \min(f(p), 2n/3)}$ and sample $G_1$ with probability proportional to $\frac{4}{\delta} \cdot V_B$.
\end{proposition}

\begin{proof}
    Consider a single iteration of the main For loop (i.e., lines 4-10 in the pseudocode in Algorithm \ref{alg:dp_estimate_1}). A point $p$ can only be sampled if $p \in \bigcup_i B_i$, and if so it is sampled with probability density proportional to $$f(p) \cdot e^{\eps \cdot 2n/3} \cdot \frac{1}{V_B} \cdot \frac{1}{3} \cdot \frac{n}{f(p)} \cdot e^{\eps (\min(f(p), 2n/3)-2n/3)} = \frac{1}{V_B} \cdot \frac{n}{3} \cdot e^{\eps (\min(f(p), 2n/3))}.$$ This is because we sample some point $i$ with $p \in B_i$ with probability density proportional to $f(p) \cdot e^{2n/3}$, and then select $p$ with density to $1/V_B$, and then accept with probability $\frac{1}{3} \cdot \frac{n}{f(p)} \cdot e^{\eps (\min(f(p), 2n/3)-2n/3)}.$ However, we select $G_1$ with probability proportional to $$4n \cdot \frac{1}{\delta} \cdot \frac{1}{3} = \frac{n}{3} \cdot \frac{4}{\delta}.$$
    So, by scaling by a factor of $\frac{3}{n} \cdot V_B$ for all probabilities, the proportionalities are correct.
\end{proof}

Therefore, if we let $X = \infty$ and keep repeating the loop until we return a point (or $G_1$), the final point that we choose has the same distribution as the algorithm described in subsection \ref{subsec:exponential}, where the garbage bucket is $G_1$.
However, to prevent the algorithm from possibly taking exponential time, we only perform the rejection sampling $X \sim Geom(1/N)$ times.
Our main additional challenge is to show that we still preserve differential privacy, while still maintaining accuracy even with a limited number of rounds.

We say that Algorithm \ref{alg:dp_estimate_1} \emph{accepts} if the final estimate is not $G_2$.
From Proposition \ref{prop:no_G2_same_alg}, we clearly have that conditioned on Algorithm \ref{alg:dp_estimate_1} accepting, the output distribution is the same as the output distribution of the computationally inefficient algorithm. Hence, the main additional results we must prove are that first, the probability of Algorithm \ref{alg:dp_estimate_1} accepting is high if most of the points $x_1, \dots, x_n$ are concentrated, and second, the probability of Algorithm \ref{alg:dp_estimate_1} accepting does not change by much among neighboring datasets.


We now state such a result: Lemma \ref{lem:analysis_G2}, which establishes that a modified version of Algorithm \ref{alg:dp_estimate_1} accepts with high probability if most of $x_1, \dots, x_n$ are in a ball of radius $r$, and is differentially private in the worst case.

\begin{lemma} \label{lem:analysis_G2}
    Suppose that $\eps, \delta, \alpha \le \frac{1}{3}$, and $n \ge C \cdot \frac{1}{\eps} \log \frac{1}{\alpha \cdot \eps \cdot \delta}$.
    Consider a modified algorithm where we run Algorithm \ref{alg:dp_estimate_1}, but output $1$ if the algorithm above returns $G_2$ and output $0$ otherwise. Then, if at least $2n/3$ of the points $x_1, \dots, x_n$ are all contained in a ball of radius $r$, then the algorithm outputs $1$ with probability at most $\alpha$. In addition, this algorithm is $(3\eps, 3\delta)$-differentially private even in the worst case.
\end{lemma}

\begin{proof}
    By scaling, we may assume that $r = 1$.
    We first prove worst-case differential privacy. To do so, we will need the following auxiliary result, for which the proof is straightforward yet somewhat tedious.
    
\begin{proposition} \label{prop:ez_dp_bash}
    Suppose that $0 \le q, q' \le 1/2$ and $0 < \eps, \delta < \frac{1}{3}$ satisfy $e^{-\eps} \cdot q - \frac{\delta}{N} \le q' \le e^{\eps} \cdot q + \frac{\delta}{N},$ where $N$ is sufficiently large. Then, $e^{-\eps} \cdot \frac{N \cdot q}{(N-1)q+1} - \delta \le \frac{N \cdot q'}{(N-1)q'+1} \le e^{\eps} \cdot \frac{N \cdot q}{(N-1)q+1} + \delta$, and $e^{-3\eps} \cdot \frac{1-q}{(N-1)q+1} - 2\delta \le \frac{1-q'}{(N-1)q'+1} \le e^{2\eps} \cdot \frac{1-q}{(N-1)q+1} + 3\delta$.
\end{proposition}
    
\begin{proof}
    First, suppose that $q' = e^{\eps} \cdot q + \frac{\delta}{N}.$ Then,
\begin{align*}
    \frac{N \cdot q'}{(N-1) \cdot q' + 1}
    &\le \frac{N \cdot q'}{(N-1) \cdot q + 1} \\
    &= \frac{N \cdot (e^{\eps} \cdot q + \frac{\delta}{N})}{(N-1) \cdot q + 1} \\
    &= e^{\eps} \cdot \frac{N \cdot q}{(N-1) \cdot q + 1} + \frac{\delta}{(N-1) \cdot q + 1} \\
    &\le e^{\eps} \cdot \frac{N \cdot q}{(N-1) \cdot q + 1} + \delta.
\end{align*}
    Likewise, suppose that $q' = e^{-\eps} \cdot q - \frac{\delta}{N}.$ Then,
\begin{align*}
    \frac{N \cdot q'}{(N-1) \cdot q' + 1}
    &\ge \frac{N \cdot q'}{(N-1) \cdot q + 1} \\
    &= \frac{N \cdot (e^{-\eps} \cdot q - \frac{\delta}{N})}{(N-1) \cdot q + 1} \\
    &= e^{-\eps} \cdot \frac{N \cdot q}{(N-1) \cdot q + 1} - \frac{\delta}{(N-1) \cdot q + 1} \\
    &\ge e^{-\eps} \cdot \frac{N \cdot q}{(N-1) \cdot q + 1} - \delta.
\end{align*}
    The first part of the proposition now follows from the fact that $\frac{N \cdot x}{(N-1) \cdot x + 1}$ is a strictly increasing function on the range $\left[-\frac{1}{N}, \infty\right)$.

    For the second part of the proposition, we first suppose that $q' = e^{\eps} \cdot q + \frac{\delta}{N}.$ In this case, note that $(N-1)q' \le e^{\eps} \cdot (N-1) \cdot q + \delta$, and that $1-q' = 1-(e^{\eps} q + \delta) = (1-q)-(e^{\eps}-1) q - \delta \ge (2-e^{\eps}) (1-q) - \delta,$ where the last inequality follows from $q \le \frac{1}{2}$ so $1-q \ge q$. Then, as long as $\eps \le \frac{1}{3}$, $2-e^{\eps} \ge e^{-2 \eps}$, so $1-q' \ge e^{-2\eps} (1-q)-\delta$. Thus,
\begin{align*}
    \frac{1-q'}{(N-1) q' + 1} 
    &\ge \frac{e^{-2\eps} \cdot (1-q) - \delta}{e^{\eps} \cdot (N-1) \cdot q + \delta + 1} \\
    &\ge \frac{e^{-2\eps} \cdot (1-q)}{e^{\eps} \cdot (N-1) \cdot q + 1} - 2 \delta\\
    &\ge e^{-3 \eps} \cdot \frac{1-q}{(N-1) \cdot q + 1} - 2 \delta.
\end{align*}
    Above, the second line follows from the fact that $\frac{a}{b} - \frac{a-\delta}{b+\delta} = \frac{(a+b)}{b(b+\delta)} \cdot \delta \le \frac{a+b}{b^2} \cdot \delta \le 2 \delta$ whenever $a \le 1 \le b$.
    Similarly, suppose that $q' = e^{-\eps} \cdot q - \frac{\delta}{N}$. Then, $(N-1)q' \ge e^{-\eps} \cdot (N-1) \cdot q - \delta$ and $1-q' = 1-(e^{-\eps}q-\delta)=(1-q)+(1-e^{-\eps})q-\delta \le (2-e^{-\eps}) \cdot (1-q) + \delta \le e^{\eps} \cdot (1-q) + \delta$. Thus,
\begin{align*}
    \frac{1-q'}{(N-1) q' + 1} 
    &\le \frac{e^{\eps} \cdot (1-q) + \delta}{e^{-\eps} \cdot (N-1) \cdot q - \delta + 1} \\
    &\le \frac{e^{\eps} \cdot (1-q)}{e^{-\eps} \cdot (N-1) \cdot q + 1} + 3 \delta \\
    &\le e^{2\eps} \cdot \frac{1-q}{(N-1) \cdot q + 1} + 3 \delta.
\end{align*}
    Above, the second line follows from the fact that $\frac{a+\delta}{b-\delta} - \frac{a}{b} = \frac{(a+b)}{b(b-\delta)} \cdot \delta \le \frac{a+b}{2/3 \cdot b^2} \cdot \delta \le 3 \delta$ whenever $a \le 1 \le b$ and $\delta \le \frac{1}{2}$.
    Thus, the second part of the proposition follows from the fact that $\frac{1-x}{(N-1) \cdot x + 1}$ is a strictly decreasing function on the range $\left[-\frac{1}{N}, \infty\right)$.
\end{proof}

    We now return to proving Lemma \ref{lem:analysis_G2}.
    First, suppose that a single round of the algorithm accepts with probability $q$. Then, note that $\BP[Geom(q) \le Geom(r)] = \frac{q}{q+r-qr}$ for any $q, r > 0$, so by setting $r =1/N$, the probability of the modified algorithm outputting $0$ is $\frac{q}{q+1/N-q/N} = N \cdot \frac{q}{(N-1)q + 1},$ and the probability of the modified algorithm outputting $1$ is $\frac{1-q}{(N-1)q+1}$.
    
    We first prove worst-case privacy. For simplicity, define $g(p) = \min(f(p), 2n/3)$. For a single round of rejection sampling, conditioned on not selecting $i = n+1$ (for which the probability is independent of the dataset $x_1, \dots, x_n$), the probability of acceptance is $$\frac{1}{3} \cdot \frac{1}{V_B} \cdot \int_{\BR^d} e^{\eps \cdot (g(p)-2n/3)} \cdot \textbf{1}_{f(p) \ge 1}.$$ So, changing one of the points $x_1, \dots, x_n$ does not affect $e^{\eps(f(p))}$ by more than a $e^{\pm \eps}$ factor, except for points that have $f(p)$ go from $1$ to $0$ or vice versa. Then, as long as $n \ge C \cdot \frac{1}{\eps} \cdot \log \frac{1}{\alpha \cdot \delta}$ for a sufficiently large constant $C$, $e^{\eps (g(p)-2n/3)} \le e^{\eps(1-2n/3)} \ll \frac{\alpha \cdot \delta}{n} \ll \frac{\delta}{N}$ for any point $p$ with $f(p) \le 1$, where we are assuming WLOG that $\eps \ge \delta$. Since the total volume of points $p$ with $f(p)$ going from $1$ to $0$ is at most $\text{Vol}(B)$ (and similar for $0$ to $1$), overall we have that if $q_{\textbf{x}}$ is the probability of acceptance in a single round of rejection sampling for $\textbf{x} = x_1, \dots, x_n,$ then if $\textbf{x}, \textbf{x'}$ differ in at most one point then $q_{\textbf{x}} \cdot e^{- \eps} - \frac{\delta}{N} \le q_{\textbf{x'}} \le q_{\textbf{x}} \cdot e^{\eps} + \frac{\delta}{N}.$
    
    Therefore, by Proposition \ref{prop:ez_dp_bash}, if $p_0$ is the probability that the modified algorithm outputs $0$ on input $\textbf{x}$ and $p_0'$ is the probability that the modified algorithm outputs $0$ on input $\textbf{x}'$, then $e^{-\eps} \cdot p_0 - \delta \le p_0' \le e^{\eps} \cdot p_0 + \delta$. Likewise, if $p_1$ is the probability that the modified algorithm outputs $1$ on input $\textbf{x}$ and $p_1'$ is the probability that the modified algorithm outputs $1$ on input $\textbf{x}'$, then $e^{-3\eps} \cdot p_1 - 2\delta \le p_1' \le e^{2\eps} \cdot p_1 + 3\delta$. (Note that $q_{\textbf{x}}, q_{\textbf{x}'} \le \frac{1}{2}$ since we never accept with probability more than $1/2$, so we may apply Proposition \ref{prop:ez_dp_bash}.) Thus, our modified algorithm is $(3\eps, 3\delta)$-differentially private, even for worst-case datasets.
    
    Next, we prove that the modified algorithm outputs $1$ with probability at most $\alpha$, assuming that at least $2/3$ of the points $x_1, \dots, x_n$ are all contained in a ball of radius $1$. In this case, if $n \ge C \cdot \frac{1}{\eps} \log \frac{1}{\alpha \cdot \eps \cdot \delta}$, then $e^{\eps \cdot 2n/3} \ge (\alpha \cdot \eps \cdot \delta)^{-1} \ge \frac{4}{\delta}$, so a single round of the rejection sampling algorithm first chooses some point $i \in [n]$ with probability at least $\frac{1}{2}$, so we select a point $i$ with $\|x_i-x\|_2 \le 1$ with probability at least $\frac{1}{3}$. Conditioned on this event, we select a point $p \in \bigcap_{i: \|x_i-x\|_2 \le 1} B_i$ with probability at least $\frac{1}{2} \cdot e^{-10 \sqrt{\log n}}$ by Lemma \ref{lem:volume}. Finally, conditioned on this event, we have that $g(p) = 2n/3,$ so we accept with probability $\frac{1}{3} \cdot \frac{n}{f(p)} \cdot e^{g(p)-2n/3} \ge \frac{1}{3}$. Therefore, each round of rejection sampling succeeds with at least $q = \frac{1}{18 \cdot e^{10 \sqrt{\log n}}}$ probability, so the probability of the modified algorithm outputting $1$ is at most $\frac{1-q}{(N-1)q+1} \le \frac{1}{N \cdot q} \le \alpha.$
\end{proof}

We now prove Theorem \ref{thm:main}.

\begin{proof}[Proof of Theorem \ref{thm:main}]
    By scaling, we may assume WLOG that $r = 1$, and we may also assume that $\eps \ge \delta$, for the same reason as in Theorem \ref{thm:exponential_main}. In addition, we may assume that $\alpha \ge \delta$. To see why, suppose we have an algorithm succeeding with $1-\delta$ probability if at least $2/3$ of the points $x_1, \dots, x_n$ are in a ball of radius $1$. Now, in $O(nd)$ time we can verify whether the coordinate-wise median of $x_1, \dots, x_n$ is within $1$ of at least $2/3$ of the points $x_1, \dots, x_n$. In this specific case, we condition on the algorithm returning a point within $O(\sqrt{d})$ of the coordinate-wise median. This happens with $1-\delta$ probability so we can repeat the algorithm until we get such a point, which occurs an expected $O(1)$ times, so the expected runtime is the same. In addition, this changes the output distribution by at most $\delta$ in total variation distance, so the algorithm is still $(O(\eps), O(\delta))$-differentially private. Finally, if at least $2/3$ of the points $x_1, \dots, x_n$ are in a ball of radius $1/10$, Theorem 1 in Narayanan~\cite{1centerclustering} proves that the coordinate-wise median of $x_1, \dots, x_n$ is within $1$ of those points, which means that this modified algorithm in fact succeeds with probability $1$. By scaling the algorithm by a factor of $10$, we have the algorithm now succeeds with probability $1$ if at least $2/3$ of the points $x_1, \dots, x_n$ are in a ball of radius $1$.
    
    Because of our assumptions, we may assume $n \ge C \cdot \frac{1}{\eps} \log \frac{1}{\eps \cdot \delta \cdot \alpha},$ since $\alpha, \eps \ge \delta$.
    
    First, we show that our algorithm successfully finds a point within $O(\sqrt{d})$ of $x$ with probability at least $1-2\alpha$, if at least $2/3$ of the points $x_i$ are contained in a ball of radius $1$. To see why, note that the overall distribution of our output is the same as first running the modified algorithm in Lemma \ref{lem:analysis_G2}, and if the output is $1$ returning $G_2$ and otherwise running the algorithm of subsection \ref{subsec:exponential}. The probability of outputting $G_2$ is at most $\alpha$ by Lemma \ref{lem:analysis_G2}, and conditioned on not outputting $G_2$, the probability of success is at least $1-\alpha$ by Proposition \ref{prop:exponential_accurate}. Therefore, the algorithm succeeds with probability at least $1-2\alpha$.
    
    Next, we bound the expected runtime even for worst-case datasets. Note that in each round of the algorithm, we pick a point $i \in [n+1]$, which can be done in $O(n)$ time\footnote{It can be done faster, but $O(n)$ is sufficient.}. We then have to pick a uniform point $p$ in the ball of radius $\sqrt{d}$ around $x_i$ if we choose $i \le n$, which is well-known to be doable in $O(d)$ time\footnote{One way of doing this is sampling a $d$-dimensional Gaussian and normalizing to $\ell_2$ norm $1$, then scaling the norm to $x \in [0, 1]$ with density proportional to $x^{d-1}$.}. To decide whether we wish to accept or reject, we need to compute $f(p)$, which takes time $O(n \cdot d)$ since we need to compute the distances between $p$ and each point $x_i$ for $1 \le i \le d.$ So overall, a single step of the rejection sampling algorithm takes time $O(n \cdot d)$. We repeat this process an expected $O(N)$ times, so the total runtime is $O(N \cdot n \cdot d)$. As a note, the number of times we repeat this loop decays geometrically, so our runtime is $O(N \cdot n \cdot d \cdot K)$ with probability at least $1-e^{\Theta(k)}$.
    
    Finally, to check privacy, we recall that our algorithm can be thought of as first running the modified algorithm in Theorem \ref{thm:exponential_main}, and if the output is $0$, running the algorithm of subsection \ref{subsec:exponential}. These procedures are $(3\eps, 3 \delta)$- and $(\eps, \delta)$-differentially private (by Lemma \ref{lem:analysis_G2} and Proposition \ref{prop:exponential_dp}, respectively). Therefore, by weak composition of differential privacy, we have that the overall algorithm is $(4 \eps, 4 \delta)$-differentially private.
\end{proof}

We now prove Corollary \ref{cor:main}.

\begin{proof}
    The algorithm simply runs $O\left(\log \frac{1}{\alpha}\right)$ parallel copies of Algorithm \ref{alg:dp_estimate_1} but with the failure probability $\alpha$ replaced by $1/3$. This gives us $k = O\left(\log \frac{1}{\alpha}\right)$ points $p_1, \dots, p_k$, each generated by an $(\eps', \delta')$-differentially private algorithm with failure probability $1/3$, according to Theorem \ref{thm:main}. Here, we will set $\eps' = \Theta(\eps)/\min\left(\log \frac{1}{\alpha}, \sqrt{\log \frac{1}{\alpha} \cdot \log \frac{1}{\delta}}\right)$ and $\delta' = \Theta(\delta)/\left(\log \frac{1}{\alpha}\right)$. The final estimate $p$ will simply be the coordinate-wise median of the points $p_1, \dots, p_k$ (where we discard any $p_i$ which is a garbage bucket $G_1$ or $G_2$). 
    
    First, by Theorem \ref{thm:main} and our new values $\eps', \delta'$, the sample complexity is 
$$n = O((\eps')^{-1} \log (\delta')^{-1}) = O\left(\frac{1}{\eps} \cdot \min\left(\log \frac{1}{\alpha}, \sqrt{\log \frac{1}{\alpha} \cdot \log \frac{1}{\delta}}\right) \cdot \log \left(\frac{1}{\delta} \cdot \log \frac{1}{\alpha}\right)\right).$$ 
    We note that this is always at most $O(\frac{1}{\eps} \log \frac{1}{\delta} \log \frac{1}{\alpha}),$ since $\log\left(\frac{1}{\delta} \log \frac{1}{\alpha}\right) = O(\log \frac{1}{\delta})$ unless $\frac{1}{\alpha} \ge e^{1/\delta}$, in which case $\sqrt{\log \frac{1}{\alpha} \log \frac{1}{\delta}} \cdot \log \left(\frac{1}{\delta} \log \frac{1}{\alpha}\right) = \tilde{O}\left(\sqrt{\log \frac{1}{\alpha}}\right) = O\left(\log \frac{1}{\alpha}\right)$.
    
    Next, the runtime is simply $k$ times the runtime of generating a single point $p_i$, since coordinate-wise medians can be computed in $O(nd)$ total time. Thus, by Theorem \ref{thm:main}, the runtime is $O\left(n e^{\sqrt{10 \log n}} \cdot d \cdot \log \frac{1}{\alpha}\right).$
    
    To bound privacy, note that since each $p_i$ is independently generated by an $(\eps', \delta')$ differentially private algorithm, and the final point $p$ only depends on the $p_i$'s, we may use either weak or strong composition to directly get the privacy guarantees of either $\left(O\left(\eps' \cdot \log \frac{1}{\alpha}\right), O\left(\delta' \cdot \log \frac{1}{\alpha}\right)\right)$ or $\left(O\left(\eps' \cdot \sqrt{\log \frac{1}{\alpha} \cdot \log \frac{1}{\delta}} + (\eps')^2 \cdot \log \frac{1}{\alpha}\right), O\left(\delta' \cdot \log \frac{1}{\alpha} + \delta\right)\right)$. Consequently, by setting $\delta' = \Theta(\delta)/\left(\log \frac{1}{\alpha}\right)$ and $\eps' = \Theta(\eps)/\min\left(\log \frac{1}{\alpha}, \sqrt{\log \frac{1}{\alpha} \cdot \log \frac{1}{\delta}}\right),$ we obtain an $(\eps, \delta)$-differential privacy bound.
    
    Finally, if at least $2n/3$ of the points $x_1, \dots, x_n$ are in some ball of radius $r$ centered at some unknown $x$, then each $p_i$ is within $O(\sqrt{d} \cdot r)$ of $x$ with probability at least $2/3$. Therefore, by a Chernoff bound, at least $3/5$ of the points $p_1, \dots, p_k$ are within $O(\sqrt{d} \cdot r)$ of $x$ with probability at least $1-\alpha$, assuming $k \ge C \log \frac{1}{\alpha}$ for a sufficiently large constant $C$. Conditioned on this, it is known that the coordinate-wise median is deterministically within $O(\sqrt{d} \cdot r)$ of $x$ as well \cite{1centerclustering}, so the algorithm succeeds with probability at least $1-\alpha$.
\end{proof}

Finally, we provide the pseudocode corresponding to Corollary \ref{cor:main}, in Algorithm \ref{alg:dp_estimate_2}.

\begin{figure}
\centering
\begin{algorithm}[H]
    \caption{: \textbf{\textsc{DP-Estimate-2}}($x_1, \dots, x_n, \alpha, \eps, \delta$).} \label{alg:dp_estimate_2}
    \begin{algorithmic}[1] 
            \State Set $k = O(\log \frac{1}{\alpha})$.
            \State Set $\eps' = \Theta(\eps)/\min\left(\log \frac{1}{\alpha}, \sqrt{\log \frac{1}{\alpha} \cdot \log \frac{1}{\delta}}\right)$, $\delta' = \Theta(\delta)/\left(\log \frac{1}{\alpha}\right)$
            \For{$i=1$ to $k$}
                \State $p_i = \textsc{DP-Estimate-1}(x_1, \dots, x_n, \frac{1}{3}, \eps', \delta')$.
            \EndFor
            \State \textbf{Return} the coordinate-wise median of $p_1, \dots, p_k$.
    \end{algorithmic}
\end{algorithm}
\caption{Differentially private estimation algorithm taking $n = O(\frac{1}{\eps} \log \frac{1}{\delta} \log \frac{1}{\alpha})$ samples, with failure probability $\alpha$ and runtime only $O(n^{1+o(1)} \cdot d \cdot \log \frac{1}{\alpha})$.}
\label{fig:dp_estimate_2}
\end{figure}

\subsection{Omitted Proofs from Subsection \ref{subsec:interpolation}} \label{subsec:interpolation_proofs}

In this subsection, we prove Theorem \ref{thm:interpolation}.
First, we describe the ``Fast Johnson-Lindenstrauss'' random rotation procedure \cite{fastjl}. First, note that we may assume WLOG that $d$ is a power of $2$, by replacing $d$ with the largest power of $2$ greater than $d$ if necessary. Now, the random rotation procedure selects a diagonal matrix $D \in \BR^{d \times d}$, where each diagonal entry $D_{ii}$ is i.i.d. selected uniformly from $\{-1, 1\}$. In addition, we also define $H \in \BR^{d \times d}$ to be the $d$-dimensional Hadamard matrix. (This matrix is only defined when $d$ is a power of $2$, which is why we need this assumption). The final random rotation procedure takes a vector $v$ and multiplies it by the matrix $M := \frac{1}{\sqrt{d}} \cdot H \cdot D \in \BR^{d \times d}$. In our setting, we will use the same random matrix $D$ for every point $x_1, \dots, x_n \in \BR^d$. It is well known \cite{fastjl} that $M$ is an invertible linear map, and importantly preserves $\ell_2$ norms, i.e., $\|Mx\|_2 = \|x\|_2$ for all $x \in \BR^d$.

Next, we formally state the main lemma for the Fast Johnson-Lindenstrauss method.

\begin{lemma} \label{lem:fastjl} \cite{fastjl}
    Suppose that a vector $v$ has length at most $1$. Then, after performing a random Fast Johnson-Lindenstrauss rotation on $v$, we obtain a vector $v' = \frac{1}{\sqrt{d}} \cdot H \cdot D \cdot v$ such that $\|v'\|_2 = \|v\|_2$ and each coordinate of $v'$ has magnitude bounded by $O(\sqrt{\log (d/\alpha)/d})$ (i.e., $\|v'\|_\infty < O(\sqrt{\log (d/\alpha)/d})$ with probability at least $1-\alpha$). In addition, multiplication by $H \cdot D$ and by $(H \cdot D)^{-1}$ can be done in $O(d \log d)$ time.
\end{lemma}

We now prove Theorem \ref{thm:interpolation}.

\begin{proof}[Proof of Theorem \ref{thm:interpolation}]
    Again, we may assume WLOG that $r = 1$.
    Set $\alpha' = \alpha/(n^2 \cdot d)$, and fix a random matrix $M := \frac{1}{\sqrt{d}} \cdot H \cdot D$, where $H$ is the Hadamard matrix and $D$ is a random diagonal matrix. In addition, fix $1 \le k \le \sqrt{d}$ as a power of $2$, and set $\eps' = \Theta\left(\eps/(k \sqrt{\log \frac{1}{\delta}})\right)$ and $\delta' = \Theta(\delta/k^2)$. For each point $x_i$ and each $1 \le s \le k^2$, we define $x_i^{(s)}$ to be the projection of $M \cdot x_i$ onto coordinates $(s-1) \cdot \frac{d}{k^2} + 1, \dots, s \cdot \frac{d}{k^2}.$ Likewise, define $x^{(s)}$ to be the projection of $M \cdot x$ onto coordinates $(s-1) \cdot \frac{d}{k^2} + 1, \dots, s \cdot \frac{d}{k^2},$ where we recall that at least $2/3$ of the points $x_i$ are within $1$ of $x$.
    
    Suppose that $n \ge C \cdot \frac{1}{\eps'} \cdot \log \frac{1}{\delta'} = \Theta(C) \cdot \frac{k}{\eps} \log \frac{k}{\delta} \sqrt{\log \frac{1}{\delta}}$. Then, by Lemma \ref{lem:fastjl} and the union bound, we obtain that if at least $2/3$ of the points $x_i$ are within $1$ of some point $x$ in Euclidean distance, then with probability at least $1-\alpha$, each point $x_i$ with $\|x_i-x\|_2 \le 1$ satisfies $\|M x_i - M x\|_\infty \le O(\sqrt{\log (d/\alpha')/d}) = O(\sqrt{\log (dn/\alpha)/d})$. Therefore, $\|x_i^{(s)}-x^{(s)}\|_2 \le O(\sqrt{\log (dn/\alpha)/d} \cdot \sqrt{d/k^2}) = O(\sqrt{\log (dn/\alpha)}/k)$ for at least $2/3$ of the points $i$ and all $s \in [k^2]$, with probability at least $1-\alpha$.
    
    Fix $s \in [k^2]$, and suppose that $\|x_i^{(s)}-x^{(s)}\|_2 \le O(\sqrt{\log (dn/\alpha)}/k)$ for at least $2/3$ of the points $i$.
    Since $n \ge C \cdot \frac{1}{\eps'} \cdot \log \frac{1}{\delta'}$, we can use Theorem \ref{thm:main} to obtain an $(\eps', \delta')$-differentially private algorithm on $\{x_1^{(s)}, \dots, x_n^{(s)}\}$ that provides an estimate $\tilde{x}^{(s)}$ of $x^{(s)}$ up to error $O(\sqrt{d/k^2}) \cdot O(\sqrt{\log (dn/\alpha)}/k) = O(\sqrt{d \log (dn/\alpha)}/k^2)$, since each $x_i^{(s)}$ lies in $d/k^2$ dimensions. So, if we perform Theorem \ref{thm:main} with failure probability $\alpha/d$ each time, then with probability $1-2\alpha$, the concatenation $\tilde{x} = (\tilde{x}^{(1)}, \dots, \tilde{x}^{(k^2)})$ satisfies $\|\tilde{x}-Mx\|_2 \le \sqrt{k^2} \cdot O(\sqrt{d \log (dn/\alpha)}/k^2) = O(\sqrt{d \log \frac{dn}{\alpha}}/k)$. Thus, since $M$ is an invertible and $\ell_2$-norm preserving linear map, so is $M^{-1}$, which means $\|M^{-1} \tilde{x}-x\|_2 \le O(\sqrt{d \log \frac{dn}{\alpha}}/k)$. Hence, if our overall algorithm outputs $M^{-1} \cdot \tilde{x}$, our algorithm is accurate up to error $O(\sqrt{d \log \frac{dn}{\alpha}}/k)$.
    
    Next, we verify that the overall algorithm is $(\eps, \delta)$-differentially private. To see this, note that since $D$ is oblivious to (i.e., has no dependence on) the dataset, multiplying the data by $M$ at the beginning and by $M^{-1}$ at the end does not degrade privacy. So, we can bound the privacy using the strong composition theorem, as a concatenation of $k^2$ $(\eps', \delta')$-differentially private functions. Since $\eps' = \Theta\left(\eps/(k\sqrt{\log \frac{1}{\delta}})\right)$ and $\delta' = \Theta(\delta/k^2)$, the strong composition theorem tells us the overall algorithm is $\left(\eps' \cdot \sqrt{2k^2 \log \frac{1}{\delta}} + O(k^2 \cdot (\eps')^2), k^2 \delta' + \delta\right) = (O(\eps), O(\delta))$-differentially private.
    
    Next, we bound the runtime of this algorithm. First, initializing $D$ and multiplying each $x_i$ by $\frac{1}{\sqrt{d}} \cdot H \cdot D$ takes time $O(d \log d \cdot n)$. Next, constructing the points $\{x_i^{(s)}\}$ takes $O(nd)$ time. Next, we apply $k^2$ private algorithms, each in $O(d/k^2)$ space, so Theorem \ref{thm:main} tells us we can do this in time $O(n^{1+o(1)} (d/k^2) \cdot \alpha^{-1}) \cdot k^2 = O(n^{1+o(1)} d \cdot \alpha^{-1})$. Finally, we concatenate our estimates $\tilde{x}^{(s)}$ and apply $M^{-1}$, which takes time $O(d \log d)$. Overall, the runtime is $O(n^{1+o(1)} d \cdot \alpha^{-1} + n d \log d)$.
    
    Now, if we assume $n = \Theta\left(\frac{k}{\eps} \cdot \log \frac{k}{\delta} \cdot \sqrt{\log \frac{1}{\delta}}\right),$ then the error is $e = O(\sqrt{d \log \frac{dn}{\alpha}}/k)$ with probability $1-\alpha$, so $n \cdot e \le O\left(\frac{1}{\eps} \cdot \log \frac{k}{\delta} \cdot \sqrt{\log \frac{1}{\delta}} \cdot \sqrt{d} \cdot \sqrt{\log \frac{dn}{\alpha}}\right) \le O\left(\sqrt{d} \cdot \frac{1}{\eps} \cdot \log \frac{d}{\delta} \cdot \sqrt{\log \frac{1}{\delta}} \cdot \sqrt{\log \frac{dn}{\alpha}}\right)$.
    This proves the algorithm is sufficiently accurate as long as $C \cdot \frac{1}{\eps} \left(\log \frac{1}{\delta}\right)^{3/2} \le n \le \sqrt{d} \cdot \frac{1}{\eps} \left(\log \frac{1}{\delta}\right)^{3/2},$ since we can always choose such a $k$ between $1$ and $\sqrt{d}$, and round it to a power of $2$ if necessary.
    
    Finally, if $C \cdot \frac{1}{\eps} \log \frac{1}{\delta} \le n \le C \cdot \frac{1}{\eps} \left(\log \frac{1}{\delta}\right)^{3/2}$, we can directly use Theorem \ref{thm:main} to get that the error is at most $O(\sqrt{d})$, which is at most $O\left(\sqrt{d} \cdot \frac{1}{\eps} \cdot \left(\log \frac{1}{\delta}\right)^{3/2}/n\right)$. So, the algorithm is also sufficiently accurate in this case.
\end{proof}

\subsection{Omitted Proofs and Pseudocode from Section \ref{sec:user}} \label{subsec:user_proofs}

In this subsection, we prove Theorem \ref{thm:user_main} and Corollary \ref{cor:user_main}, and provide pseudocode for Theorem \ref{thm:user_main}.

\begin{proof}[Proof of Theorem \ref{thm:user_main}]
    First, let us assume that no user is corrupted. In this case, a fixed user $i$ is given $m$ data points $X_{i, 1}, \dots, X_{i, m},$ and computes the mean of their data, which we denote $\bar{X}_i$. Each $X_{i, j}$ has mean $\mu$, and
\begin{align*}
    \BE\left[\|\bar{X}_i-\mu\|_2^2\right] &= \frac{1}{m^2} \left(\sum_{j = 1}^{m} \BE\left[\|X_{i, j}-\mu\|_2^2\right] + \sum_{j \neq j'} \BE\langle X_{i, j}-\mu, X_{i, j'}-\mu \rangle\right) \\
    &\le \frac{1}{m^2} \cdot \left(m \cdot r^2 + (m)(m-1) \cdot \frac{r^2}{m}\right) \\
    &\le \frac{2}{m} \cdot r^2.
\end{align*}

So, by Markov's inequality, $\|\bar{X}_i-\mu\|_2 \le \frac{10}{\sqrt{m}} \cdot r$ with probability at least $\frac{49}{50},$ so by a Chernoff bound, at least $\frac{11}{12} \cdot n$ of the users will have $\|\bar{X}_i-\mu\|_2 \le \frac{10}{\sqrt{m}} \cdot r$ with at least $1-e^{\Omega(n)} \ge 1-\alpha$ probability. Therefore, even if an adversary looks at the users' data and arbitrarily corrupts $n/4$ of the users' data, at least $\frac{2}{3} \cdot n$ of the users will still have $\|\bar{X}_i-\mu\|_2 \le \frac{10}{\sqrt{m}} \cdot r$.

Under this event, by replacing $r$ with $r' = \frac{10}{\sqrt{m}} \cdot r$, we can directly apply Theorem \ref{thm:main} or \ref{thm:interpolation} to obtain the desired result with runtime $O(mnd + n^{1+o(1)} \cdot d \cdot \alpha^{-1} + nd \log d)$, since computing each $\bar{X}_i$ takes time $O(mnd)$. Alternatively, we could apply Corollary \ref{cor:main} or the remark after Theorem \ref{thm:interpolation} to obtain the improved runtime of $O(mnd + n^{1+o(1)} \cdot d \cdot \log \frac{1}{\alpha} + nd \log d)$, at the cost of the number of users $n$ being an additional $\log \frac{1}{\alpha}$ factor larger.
\end{proof}

\begin{proof}[Proof of Corollary \ref{cor:user_main}]
    First, we consider the mean estimation question. In the case where $C \cdot \frac{1}{\eps} \log \frac{1}{\alpha \cdot \delta} \le n \le C \cdot \sqrt{d} \cdot \frac{1}{\eps} \log \frac{1}{\delta},$ we can directly apply Theorem \ref{thm:user_main}. Indeed, since all points $X_{i, j}$ are drawn from a distribution $\mathcal{D}$ supported on a ball of radius $r$, we have that $\BE\left[\|X_{i, j}-\mu\|_2^2\right] \le \BE\left[\|X_{i, j}-x\|_2^2\right] \le r^2$, where $r$ is the center of the ball. In addition, since each $X_{i, j}$ is independent, we have that $\BE\left[\langle X_{i, j}-\mu, X_{i, j'}-\mu\rangle \right] = \langle \BE[X_{i, j}-\mu], \BE[X_{i, j'}-\mu]\rangle = 0$. Therefore, the conditions to apply Theorem \ref{thm:user_main} hold.
    
    In the case where $n \ge C \cdot \sqrt{d \log \frac{1}{\delta}} \cdot \frac{1}{\eps} \cdot \log (m(dn+n^2 \eps^2))$, we can directly apply Corollary 1 in Levy et al.~\cite{levy2021user} to obtain the desired bound (where we ignore all logarithmic factors). Finally, in the potential intermediate regime where $C \cdot \sqrt{d} \cdot \frac{1}{\eps} \log \frac{1}{\delta} \le n \le C \cdot \sqrt{d \log \frac{1}{\delta}} \cdot \frac{1}{\eps} \cdot \log (m(dn+n^2 \eps^2)),$ we can simply ignore all users after the first $C \cdot \sqrt{d} \cdot \frac{1}{\eps} \log \frac{1}{\delta}$, and use our already existing bounds. Since the upper bound is at most an $O(\log (m \cdot d \cdot n))$ factor more than the lower bound, the result still holds in this case since we are ignoring logarithmic factors in the bound.
    
    We now move to our result on the distribution learning question. This will be quite direct from our results on the mean estimation question. Indeed, we can encode each $i \in [d]$ as a one-hot vector $(0, \dots, 1, 0, \dots, 0) \in \BR^d,$ where there is a $1$ precisely at the $i$th position. Then, this distribution is supported on the unit ball around $0$, so we can learn this distribution up to $\ell_2$ error
\[\tilde{O}\left(\frac{\sqrt{d}}{\eps \cdot n \cdot \sqrt{m}} + \frac{1}{\sqrt{m \cdot n}}\right)\]
    with failure probability $\alpha$. However, note that the $\ell_1$ distance is at most $\sqrt{d}$ times the $\ell_2$ distance, and the total variation distance of two distributions is exactly $\frac{1}{2}$ of their $\ell_1$ distance. So, we can learn the distribution $\mathcal{D}$ up to total variation distance
\[\tilde{O}\left(\frac{d}{\eps \cdot n \cdot \sqrt{m}} + \frac{\sqrt{d}}{\sqrt{m \cdot n}}\right)\]
    in an $(\eps, \delta)$ user-level differentially private manner.
\end{proof}

Finally, we include pseudocode for the algorithm corresponding to Theorem \ref{thm:user_main}, in Algorithm \ref{alg:dp_estimate_user}.

\begin{figure}
\centering
\begin{algorithm}[H]
    \caption{: \textbf{\textsc{DP-Estimate-User}}$\left((X = \{X_{i, j}\}_{i=1, j=1}^{n \hspace{0.38cm} m}), \alpha, \eps, \delta, k\right)$.} \label{alg:dp_estimate_user}
    \begin{algorithmic}[1] 
            \For{$i=1$ to $n$}
                \State Let $\bar{X}_i$ be the mean of $X_{i, 1}, \dots, X_{i, m}$
            \EndFor
            \State Sample $D$ uniformly over $\pm 1$-valued diagonal matrices
            \For{$i=1$ to $n$}
                \State Compute $\bar{X}'_i = \frac{1}{\sqrt{d}} \cdot HD \cdot X_i$
                \For{$s=1$ to $k^2$}
                    \State Compute $\bar{X}_i^{(s)}$ as $\bar{X}'_i$ restricted to coordinates $(s-1) \cdot (n/k^2)+1, \dots, s \cdot (n/k^2)$
                \EndFor
            \EndFor
            \If{$k > 1$}
                \State $\eps' = \Theta\left(\eps/(k\sqrt{\log \frac{1}{\delta}})\right),$ $\delta' = \Theta(\delta/k^2)$
            \Else
                \State $\eps' = \eps$, $\delta' = \delta$ \Comment{When $k=1$, there is no need to use any composition.}
            \EndIf
            \For{$s=1$ to $k^2$}
                \State Compute $\tilde{X}^{(s)} = \textsc{DP-Estimate-1}(\bar{X}_1^{(s)}, \dots, \bar{X}_n^{(s)}, \alpha/k^2, \eps', \delta')$ \Comment{We may use \textsc{DP-Estimate-2} instead of \textsc{DP-Estimate-1}.}
            \EndFor
            \State Compute $\tilde{X}$ as the concatenation of all $\tilde{X}^{(s)}$
            \State \textbf{Return} $\sqrt{d} \cdot (HD)^{-1} \cdot \tilde{X}$
    \end{algorithmic}
\end{algorithm}
\caption{User-Level differentially private estimation algorithm taking $O(\frac{1}{\eps} \log \frac{1}{\alpha \cdot \delta}) \le n \le O\left(\sqrt{d} \cdot \frac{1}{\eps} \log \frac{1}{\delta}\right)$ users and $m$ samples in $\BR^d$ per user, and a failure probability $\alpha$. Here, we set $k$ so that $n = \Theta\left(\frac{k}{\eps} \log \frac{k}{\delta} \sqrt{\log \frac{1}{\delta}}\right)$.}
\label{fig:dp_estimate_user}
\end{figure}

\section{Missing Proofs: Lower Bounds} \label{sec:lb_proofs}

In this section, we prove all lower bound results, from Section \ref{sec:lb}. In Subsection \ref{subsec:lb_easy_proof}, we prove Lemma \ref{lem:lb_easy}. In Subsection \ref{subsec:lb_main_proof}, we prove Theorem \ref{thm:lb_main}. (We defer one piece of this proof to Subsection \ref{subsec:complex_analysis}). Finally, in Subsection \ref{subsec:robustness_lower}, we state and prove a simple result that one cannot learn the mean of a distribution up to error better than $O\left(\frac{r}{\sqrt{m}}\right)$ if a constant fraction of users are corrupted, even if privacy is not required. This roughly implies that beyond $n \approx \sqrt{d}$ users, having more users does not help with robust mean estimation.

\subsection{Proof of Lemma \ref{lem:lb_easy}} \label{subsec:lb_easy_proof}

In this section, we prove Lemma \ref{lem:lb_easy}.

\begin{proof}
    Suppose otherwise, i.e., there is such an algorithm which we call $M$.
    Consider a ball $B$ of radius $1$ centered at some point $b$, and for any $r > 0$, let $B_r$ represent the radius $r$ ball concentric with $B$. Suppose we have data $\{X_{i, j}\}$ where each point $X_{i, j}$ is arbitrary. Let $p_0$ represent the probability that our user-level differentially private algorithm $M$ outputs a point $x \in B_T$.
    
    Now, for $1 \le i \le n,$ let $p_i$ represent the probability that $M$ outputs a point $x \in B_T$, if we replace the first $i$ users' data each with $m$ copies of the center of $B$. Then, $p_n$ is the probability that $M$ outputs $x \in B_T$ if every point $X_{i, j}$ is just the center of $B$. By our assumption on $M$, and since each point $X_{i, j}$ is in $B$, $p_n \ge \frac{2}{3}.$ In addition, we have that this algorithm is $(\eps, \delta)$ user-level differentially private, which means that $p_{i+1} \le e^{\eps} \cdot p_i + \delta$. Equivalently, by adding $\delta \cdot \frac{1}{e^{\eps}-1}$ to both sides, we have that $p_{i+1} + \delta \cdot \frac{1}{e^{\eps}-1} \le e^{\eps} \cdot p_i + \delta \cdot \frac{e^{\eps}}{e^{\eps}-1}$, so
\[\frac{p_{i+1} + \frac{\delta}{e^{\eps}-1}}{p_i + \frac{\delta}{e^{\eps}-1}} \le e^{\eps}\]
    for all $0 \le i \le n-1$. So, this means that
\[p_0 + \frac{\delta}{e^{\eps}-1} \ge e^{-n \eps} \cdot \left(p_n + \frac{\delta}{e^{\eps}-1}\right) \ge e^{-n \eps} \cdot \frac{2}{3}.\]
    If $n \le \frac{1}{3} \cdot \frac{1}{\eps} \cdot \log \frac{1}{\delta},$ then $e^{-n \eps} \cdot \frac{2}{3} \ge \frac{2}{3} \cdot (\delta)^{1/3}$, which means that $p_0 \ge \frac{2}{3} \cdot (\delta)^{1/3} - \frac{\delta}{e^{\eps}-1} \ge \delta$, assuming that $\delta \le \eps^2$ and that $\delta$ is sufficiently small.
    
    Therefore, even if the data $\{X_{i, j}\}$ is arbitrary, we still have that with probability at least $\delta$, we output a point in $B_T,$ i.e., within $T$ of $b$. But this is true for an arbitrary $b$, so by constructing $N > \frac{1}{\delta}$ disjoint balls of radius $T$ centered at some points $b_1, \dots, b_N$, which have pairwise distances at least $2T$ from each other, we have that the probability that $M$ outputs a point within $T$ of some $b_i$ is at least $\frac{N}{\delta} > 1$, which is a contradiction. Therefore, we must have that $n \ge \frac{1}{3} \cdot \frac{1}{\eps} \cdot \log \frac{1}{\delta}.$
\end{proof}

\subsection{Proof of Theorem \ref{thm:lb_main}} \label{subsec:lb_main_proof}

In this section, we prove Theorem \ref{thm:lb_main}.

Suppose we have $n$ users and $m$ samples (items) per user, where each sample is drawn from an unknown distribution over $[d]$ characterized by $\vec{p} = [p_1, \dots, p_d]$. 
The goal of the algorithm is to privately estimate $\vec{p}$ up to low total variation ($\ell_1$) distance, given these samples. Our lower bound will show that one cannot privately estimate $\vec{p}$ significantly better than the guarantees of Corollary \ref{cor:user_main}. In addition, our lower bound will also imply a nearly optimal lower bound for user-level mean estimation (up to $\ell_2$ distance), both in the case when there are few users or many users. This is because we can think of each sample $i \sim [d]$ as a one-hot vector $(0, 0, \dots, 1, 0, \dots, 0)$, with a $1$ on the $i$th coordinate. Hence, the samples are concentrated in a ball of radius $1$ and we are trying to estimate the mean $(p_1, \dots, p_d)$.

Before dealing with multinomial distributions over $[d]$, we start by just considering the number of samples that equal some $k \in [d]$ for each user. However, instead of looking at Binomial distributions, which is the actual distribution, we instead prove the following two lemmas relating to Poisson variables.
Our method of attack is to show that if instead of each user getting $m$ samples from a distribution over $d$, the user were given $d$ independent Poisson variables where the $k$th variable roughly corresponding to the number of samples that the $i$th user receive that equal $k$, then privately estimating the distribution is hard. This step follows a similar technique to Kamath et al.~\cite{kamath2019highdimensional}, but with very different computations. We then show how to reduce the problem of each user receiving samples from $[d]$ to the problem of each user receiving Poisson variables, which shows that learning discrete distributions privately is also difficult.

\begin{lemma} \label{lem:dp_lower_1}
    Fix positive integers $m, n$ and let $0 < p < 1$ be a parameter. Let $f$ be a nonnegative, uniformly bounded function on $n$ variables $X_1, \dots, X_n$, where each $X_i \overset{i.i.d.}{\sim} \text{Pois}(m \cdot p)$. For simplicity, we write $X = (X_1, \dots, X_n)$, and write $\BE_X$ to represent the expectation over drawing $X_i \overset{i.i.d.}{\sim} \text{Pois}(m \cdot p)$. Define $g(p)= \BE_{X} [f(X)].$ Then,
\begin{equation}
    p \cdot g'(p) = \BE_{X} \left[f(X) \cdot \sum_{i = 1}^{n} (X_i-mp)\right],
\end{equation}
    where $g'(p)$ is the derivative of $g$ at $p$.
\end{lemma}

\begin{proof}
    Note that by the formula for the Poisson probability mass function,
\[g(p) = \sum_{X_1, \dots, X_n \ge 0} \prod_{i = 1}^{n} f(X) \cdot \frac{e^{-m p} (m p)^{X_i}}{X_i!} = e^{-m n \cdot p} \sum_{X_1, \dots, X_n \ge 0} f(X) \cdot \frac{m^{\sum X_i} \cdot p^{\sum X_i}}{X_1! \cdots X_n!}.\]
    Now, since $|f(X)|$ is uniformly bounded, one can show that the derivative of the expression $\sum_{X_1, \dots, X_n \ge 0} f(X) \cdot \frac{m^{\sum X_i} \cdot p^{\sum X_i}}{X_1! \cdots X_n!}$ commutes with the summation. (We defer the proof of this fact to Lemma \ref{lem:derivative_commuting} in Subsection \ref{subsec:complex_analysis}.) So, using the product rule,
\begin{align*}
    g'(p) &= -mn \cdot e^{-m n \cdot p} \sum_{X_1, \dots, X_n \ge 0} f(X) \cdot \frac{m^{\sum X_i} \cdot p^{\sum X_i}}{X_1! \cdots X_n!} + e^{-mn \cdot p} \cdot \sum_{X_1, \dots, X_n \ge 0} f(X) \cdot \frac{m^{\sum X_i} \cdot (\sum X_i) \cdot p^{(\sum X_i) - 1}}{X_1! \cdots X_n!} \\
    &= -mn \cdot \BE_X[f(X)] + \frac{1}{p} \cdot \BE_X[(X_1+\cdots+X_n) \cdot f(X)] \\
    &= \frac{1}{p} \cdot \BE_X\left[f(X) \cdot \sum_{i = 1}^{n} \left(X_i-mp\right)\right]. 
\end{align*}
\end{proof}

The next lemma, sometimes called a ``fingerprinting'' lemma, is based on similar techniques used in \cite{kamath2019highdimensional, steinke2015interactive} for lower bounds for private estimation from samples.

\begin{lemma} \label{lem:dp_lower_2}
    Let $p \sim \text{Unif}\left[\frac{0.5}{d}, \frac{1.5}{d}\right]$ and let $X_1, \dots, X_n \overset{i.i.d.}{\sim} \text{Pois}(m \cdot p)$. Then, for $f$ as in Lemma \ref{lem:dp_lower_1},
\[\BE_{p, X}\left[(f(X)-p)^2 + \frac{\left(p - \frac{0.5}{d}\right) \cdot \left(\frac{1.5}{d}-p\right)}{p} \cdot (f(X)-p) \cdot \sum_{i = 1}^{n} (X_i-mp)\right] \ge \frac{1}{12 d^2}.\]
\end{lemma}

\begin{proof}
First, by Lemma \ref{lem:dp_lower_1} we have that
\[h(p) := \BE_{X_1, \dots, X_n \overset{i.i.d.}{\sim} \text{Pois}(m \cdot p)}\left[\frac{\left(p - \frac{0.5}{d}\right) \cdot \left(\frac{1.5}{d}-p\right)}{p} \cdot f(X_1, \dots, X_n) \cdot \sum_{i = 1}^{n} (X_i-mp)\right]\]
equals $g'(p) \cdot \left(p - \frac{0.5}{d}\right) \cdot \left(\frac{1.5}{d}-p\right).$
Therefore,
\begin{align*}
    \BE_{p \sim \text{Unif}[\frac{0.5}{d}, \frac{1.5}{d}]} [h(p)] &= d \cdot \int_{0.5/d}^{1.5/d} g'(p) \cdot \left(p - \frac{0.5}{d}\right) \cdot \left(\frac{1.5}{d}-p\right) dp \\
    &= 2d \cdot \int_{0.5/d}^{1.5/d} g(p) \cdot (p - 1/d) \\
    &= 2 \cdot \BE_{p \sim \text{Unif}[\frac{0.5}{d}, \frac{1.5}{d}]}[g(p) \cdot (p-1/d)],
\end{align*}
where the second line follows from integration by parts.

We abbreviate $X = (X_1, \dots, X_n)$ and assume $p \sim \text{Unif}\left[\frac{0.5}{d}, \frac{1.5}{d}\right]$ and $X_i \overset{i.i.d.}{\sim} \text{Pois}(m \cdot p)$. Then,
\begin{align*}
\BE_{p, X}[(f(X)-p)^2] &= \BE_{p, X}[((f(X) - 1/d) - (p-1/d))^2]\\
&\ge \BE_{p} [(p-1/d)^2] - 2 \BE_{p, X}[(f(X)-1/d)(p-1/d)] \\
&= 1/(12d^2) - 2 \BE_{p, X}[f(X) \cdot (p-1/d)]\\
&= 1/(12d^2) - 2 \BE_p[g(p) \cdot (p-1/d)]\\
&= 1/(12 d^2) - \BE_p[h(p)].
\end{align*}
Therefore, we have that

\begin{equation*}
    \BE_{p, X}\left[(f(X)-p)^2 + \frac{\left(p - \frac{0.5}{d}\right) \cdot \left(\frac{1.5}{d}-p\right)}{p} \cdot (f(X)-p) \cdot \sum_{i = 1}^{n} (X_i-mp)\right] \ge \frac{1}{12 d^2}. 
\end{equation*}
\end{proof}

Rather than thinking of $X = \{X_{i, j}\}$ where each $X_{i, j} \in [d]$ and $1 \le i \le n, 1 \le j \le m$, we consider $X = \{X_{i, k}\}$, where $1 \le i \le n$ and $1 \le k \le d$: here, $X_{i, k}$ represents the count of the number of times user $i$ received $k \in [d]$. The $i$th user is given $\{X_{i, k}\}_{k = 1}^{d},$ but instead of it following the multinomial distribution $\text{Mult}(m; \vec{p})$ for $\vec{p} = (p_1, \dots, p_d)$ we instead suppose that $X_{i, k} \sim \text{Pois}(m \cdot p_i)$ for each $i$, where $m$ is fixed, the values $X_{i, k}$ are independent (both across $i$ and $k$). Finally, we assume that each $p_i \overset{i.i.d.}{\sim} \text{Unif}\left[\frac{0.5}{d}, \frac{1.5}{d}\right],$ and that the algorithm does not know $\vec{p} = (p_1, \dots, p_k)$. 

We now prove a lower bound on the sample complexity if the users are given samples from a Poisson distribution, rather than a multinomial or discrete distribution. This follows a similar method to Theorem 6.1 in Kamath et al.~\cite{kamath2019highdimensional}

\begin{lemma} \label{lem:similar_to_kamath}
    Suppose that $M$ is an algorithm that takes as input $X = \{X_{i, k}\}$, and with all notation as in the above paragraphs, and outputs an estimate $M(X) = (M_1(X), \dots, M_d(X))$, where each $M_k(X) \in [\frac{0.5}{d}, \frac{1.5}{d}]$. Suppose that $M$ is $(\eps, \delta)$ user-level differentially private for $\delta = o(\frac{\eps \cdot \alpha}{d^2 \cdot \sqrt{m}})$ (where the $i$th user is given $\{X_{i, k}\}_{k = 1}^{d}$), and that $\BE_{\vec{p}, X, M}\left[\|M(X)-\vec{p}\|_2^2\right] \le \frac{\alpha^2}{24 d}$ for some $0 < \alpha \le 1$, where the expectation is also over the randomness of the algorithm $M$. Then, the number of users $n$ must satisfy $n \ge \Omega(\frac{d}{\alpha \cdot \eps \cdot \sqrt{m}})$.
\end{lemma}

\begin{proof}
    Define $Z_{i, k} = \left(\frac{(p_k-0.5/d)(1.5/d-p_k)}{p_k}\right) \cdot (M_k(X)-p_k) \cdot (X_{i, k}-mp_k)$, and define $Z_i = \sum_{k = 1}^{d} Z_{i, k}$. Now, by Lemma \ref{lem:dp_lower_2}, we have that for any $k \in [d]$,
\[\BE_{\vec{p}, X, M}\left[(M_k(X)-p_k)^2 + \sum_{i = 1}^{n} Z_{i, k}\right] \ge \frac{1}{12 d^2}.\]
    While this implication may not seem immediate because $M_k(X)$ does not just have to depend on the choices $X_{i, k}$ for a fixed $k$, we note that since we chose the $p_k$ values independently, the set $\{X_{i, k'}\}$ over $1 \le i \le n$ and $k' \neq k$ is independent of the set $\{X_{i, k}\}$ over $1 \le i \le n$. Therefore, Lemma \ref{lem:dp_lower_2} still holds in our case. Therefore, since $\BE_{\vec{p}, X, M}[\|M(X)-\vec{p}\|_2^2] \le \frac{\alpha^2}{24 d}$, we must have that
\[\sum_{i = 1}^{n} \BE_{\vec{p},X, M}[Z_i] = \sum_{k = 1}^{d} \BE_{\vec{p}, X, M}\left[\sum_{i = 1}^{n} Z_{i, k}\right] \ge \frac{1}{12d} - \BE_{\vec{p}, X, M} \left[\|M(X)-\vec{p}\|_2^2\right] \ge \frac{1}{24 d}.\]

We spend the remainder of the proof working towards an upper bound for $\sum_{i = 1}^{n} \BE[Z_i]$. Fix $\vec{p}$, and let $X_{\sim i}$ represent replacing $X_{i, \cdot}$ with an independent draw from $\text{Pois}(m \cdot p_k)$ for each $1 \le k \le d$. (In other words, we have replaced all of the $i$th user's data.) Let $\tilde{Z}_{i, k} = \left(\frac{(p_k-0.5/d)(1.5/d-p_k)}{p_k}\right) \cdot (M_k(X_{\sim i})-p_k) \cdot (X_{i, k}-mp_k)$ (only the middle term in the product is changed from $Z_{i, k}$) and let $\tilde{Z}_i = \sum_{k = 1}^{d} \tilde{Z}_{i, k}$. Due to the resampling, we have that $\{X_{i, k}\}_{k = 1}^{d}$ is independent of $X_{\sim i}$ conditioned on $\vec{p}$, so $\BE_{X, M}[\tilde{Z}_i] = 0$. Now, note that $\BE_{X, X_{\sim i}, M}[\tilde{Z}_{i, k} \cdot \tilde{Z}_{i, k'}] = 0$ for any $k \neq k'$ if we condition on $\vec{p}$, since up to a scaling factor only depending on $\vec{p}$,
$$\tilde{Z}_{i, k} \cdot \tilde{Z}_{i, k'} \propto [(M_k(X_{\sim i}) - p_k) \cdot (M_{k'}(X_{\sim i}) - p_{k'})] \cdot [(X_{i, k} - m p_k) \cdot (X_{i, k'} - m p_{k'})],$$
and the two terms are independent because $X_{\sim i}$ and $X_{i, \cdot}$ are independent. However, $(X_{i, k}-mp_k)$ and $(X_{i, k'}-mp_{k'}$) are independent and both have expectation $0$, so this means the overall product has expectation $0$.

Now, note that $\frac{(p_k-0.5/d)(1.5/d-p_k)}{p_k} = O(1/d)$ for $p_k \in [\frac{0.5}{d}, \frac{1.5}{d}],$ that $|M_k(X_{\sim i})-p_k|$ is always $O(1/d)$, and that $\BE[(X_{i, k}-mp_k)^2] = mp_k = O(m/d)$. Thus, conditioning on $\vec{p}$, we can bound
\begin{align*}
    \BE_{X, X_{\sim i}, M}[\tilde{Z}_i^2] &= \sum_{k = 1}^d \BE_{X, X_{\sim i}, M}[\tilde{Z}_{i, k}^2] + \sum_{k \neq k'} \BE_{X, X_{\sim i}, M}[\tilde{Z}_{i, k} \cdot \tilde{Z}_{i, k'}] \\
    &\le O(1/d^2) \cdot \sum_{k = 1}^{d} \BE_{X_{\sim i}, M}[(M_k(X_{\sim i}) - p_k)^2] \cdot \BE_X[(X_{i, k} - m p_k)^2] \\
    &\le O(1/d^2) \cdot O(m/d) \cdot \BE_{X_{\sim i}, M}\left[\|M(X_{\sim i})-\vec{p}\|_2^2\right] \\
    &= O(m \cdot \alpha^2/d^4).
\end{align*}
Note that in the final line, we are using the fact that $M(X_{\sim i})$ and $M(X)$ have the same distribution, and we are assuming that $\BE[\|M(X)-\vec{p}\|_2^2] \le \alpha^2/(24 d)$.
Hence, even when we remove the conditioning on $\vec{p}$, we still have that $\BE[\tilde{Z}_i^2] = O(m/d^4)$.
Also, $\BE[|\tilde{Z}_i|]^2 \le \BE[\tilde{Z}_i^2]$ by Cauchy-Schwarz, which means that $\BE_{\vec{p}, X, X_{\sim i}, M}[|\tilde{Z}_i|] \le \sqrt{\BE[\tilde{Z}_i^2]}= O(\sqrt{m} \cdot \alpha/d^2)$.

To finish, we note that if we fix $\vec{p}$, $X$, and $X_{\sim i}$ and just take the expectation over the randomness of $M$, then $Z_{i}$ and $\tilde{Z}_i$ are the same, except one is a function of $M(X)$ and the other is the same function of $M(X_{\sim i})$. Therefore, the distributions of $Z_{i}$ and $\tilde{Z}_i$ conditioned on $\vec{p}$, $X$, and $X_{\sim i}$ are $(\eps, \delta)$-close. Thus, the distributions are still $(\eps, \delta)$ close when we remove the conditioning as well. In this case, it is well known (see Equation 18 in Kamath et al.~\cite{kamath2019highdimensional}) that for any $T > 0$,
\begin{equation} \label{eq:dp_eq}
    \BE[Z_i] \le \BE[\tilde{Z}_i] + 2 \eps \cdot \BE[|\tilde{Z}_i|] + 2 \delta \cdot T + 2 \int_T^\infty \BP[Z_i > t] dt.
\end{equation}
To bound $\BP[Z_i > t]$, it is well known that for any variable $Y \sim \text{Pois}(\lambda),$ that $\BP(Y > \lambda + t) \le \exp\left(-c \cdot t\right)$ whenever $t > \lambda$, for some absolute constant $c > 0$. Therefore, the probability that $|X_{i, k}-mp_k| > t$ is at most $\exp\left(- c \cdot t\right)$ for any $t > 1.5 \cdot \frac{m}{d},$ since $p_k \in \left[\frac{0.5}{d}, \frac{1.5}{d}\right].$
We note that $Z_{i, k} \le O(1/d^2) \cdot |X_{i, k}-m p_k|$, which means that any $t > \Omega(m/d^3)$, the probability that $Z_{i, k} > t$ is at most $\exp\left(-\Omega(d^2) \cdot t\right).$ So, as $Z_i = \sum_{k = 1}^{d} Z_{i, k}$, the probability that $Z_{i, k} > t$ is at most $d \cdot \exp\left(-\Omega(d) \cdot t\right)$ for any $t > T = \Theta(m)$ by a basic union bound. Thus, we can apply \eqref{eq:dp_eq} and the fact that $\BE[\tilde{Z}_i] = 0$ and $\BE[|\tilde{Z}_i|] = O(\sqrt{m} \cdot \alpha/d^2)$ to get
\begin{align*}
    \BE[Z_i] &\le \BE[\tilde{Z}_i] + 2 \eps \cdot \BE[|\tilde{Z}_i|] + 2 \delta \cdot T + 2 \int_T^\infty \BP[Z_i > t] dt \\
    &\le O(\eps \cdot \alpha) \cdot \frac{\sqrt{m}}{d^2} + O(\delta) \cdot m + 2 \cdot \int_{\Theta(m)}^{\infty} d \cdot e^{-\Omega(d) \cdot t} dt \\
    &\le O(\eps \cdot \alpha) \cdot \frac{\sqrt{m}}{d^2} + O(\delta \cdot m) + O(e^{-\Omega(d)}).
\end{align*}
    Thus,
\[\frac{1}{24 d} \le \sum_{i = 1}^{n} \BE[Z_i] \le n \cdot \left(O(\eps \cdot \alpha) \cdot \frac{\sqrt{m}}{d^2} + O(\delta \cdot m) + O(e^{-\Omega(d)})\right).\]
    Assuming that $\delta = o(\eps \cdot \alpha/(d^2 \cdot \sqrt{m}))$, this means that $n = \Omega(\frac{1}{\eps} \cdot \alpha^{-1} \cdot d/\sqrt{m})$, as desired.
\end{proof}

We now provide a reduction from privately learning multinomial distributions to privately learning Poisson distributions as above. Given this reduction, we can prove that learning a multinomial distribution accurately also requires $n = \tilde{\Omega}(\frac{1}{\eps} \cdot \alpha^{-1} \cdot d/\sqrt{m})$ samples.

First, we note the following result about privately learning $1$-dimensional distributions, which will be useful in proving our reduction.

\begin{theorem} \cite{levy2021user} \label{thm:one_dim_estimation}
    Let $\tau < B$ and $\eps < 1$ be fixed. Then, there exists an $(\eps, 0)$-differentially private algorithm $\mathcal{D}$ with input $X_1, \dots, X_n \in [0, B]$, such that if there exists an unknown interval $I \subset [0, B]$ of radius $2 \tau$ that contains all of the $X_i$'s, then for any $T > 0$, the probability that $\left|\mathcal{D}(X_1, \dots, X_n) - \frac{X_1+\cdots+X_n}{n}\right| \ge \frac{8 \tau}{n \eps} \cdot T$ is at most $e^{-T} + \frac{B}{\tau} \cdot e^{-n \eps/8}.$
\end{theorem}

\begin{theorem} \label{thm:lowerbound_remove_poissonization}
    Suppose there exists an $(\eps, \delta)$ user-level differentially private algorithm $M$, where $\delta = o\left((\eps \cdot \alpha)/(d^2 \cdot \sqrt{m} \cdot \log^2(\frac{d}{\eps \cdot \alpha})\right)$, that takes as input $X = \{X_{i, j}\}$, where $i$ ranges from $1$ to $n$, $j$ ranges from $1$ to $m$, and $X_{i, j} \in [d]$. Moreover, suppose that if $X_{i, j} \overset{i.i.d.}{\sim} \tilde{p}$ for some distribution $\tilde{p}$ over $[d]$,
    then $\BE\left[\|M(X)-\tilde{p}\|_1^2\right] \le \frac{\alpha^2}{1000}$. Then, $n \sqrt{m}= \Omega\left(\eps^{-1} \alpha^{-1} d/\log^6\left(\frac{d}{\alpha \cdot \eps}\right)\right)$.
\end{theorem}

\begin{proof}
    First, note that $n \ge \frac{1}{6 \eps}$. To see why, since $\delta \le \eps$, for two adjacent datasets $X = \{X_{i, j}\}$ and $X' = \{X'_{i, j}\}$ that only differ on $1$ user, the total variation distance between $M(X)$ and $M(X')$ is at most $\eps + \delta \le 2 \eps.$ So, the total variation distance between $M(X)$ and $M(X')$ for two arbitrary datasets $X, X'$, which may differ in up to all $n$ users, is at most $2 n \cdot \eps$. So, if $n < \frac{1}{6 \eps},$ then the total variation distance between $M(X)$ and $M(X')$ is less than $\frac{1}{3}$. Therefore, if $X$ were drawn from the distribution where every element is $1$, and $X'$ were drawn from the distribution where every element is $2$, either $M(X)$ estimates $p_1 > \frac{1}{2}$ with at most $2/3$ probability, or $M(X')$ estimates $p_2 > \frac{1}{2}$ with at most $2/3$ probability. Therefore, the expected value of $\|M(X)-\tilde{p}\|_1^2$ is at least $\frac{1}{3} \cdot \left(\frac{1}{2}\right)^2 = \frac{1}{12}$ either for $\tilde{p} = (1, 0, \dots, 0)$ or for $\tilde{p} = (0, 1, 0, \dots, 0)$. So, we may assume $n \ge \frac{1}{6 \eps}$.

    Now, let $A = 6 C \left(\log \frac{d}{\alpha \cdot \eps}\right)^4$ for some large constant $C$, and let $n' = n \cdot A$ and $m' = m \cdot A$. It suffices to show that $n' \sqrt{m'} = \Omega(\eps^{-1} \alpha^{-1} d)$.
    Note that $n' \ge C \cdot \frac{1}{\eps} \cdot \left(\log \frac{d}{\alpha \cdot \eps}\right)^4$ and that $m' \ge C \cdot \left(\log \frac{d}{\alpha \cdot \eps}\right)^4$.
    Finally, we may assume that both $n', m' \le O\left(\frac{d}{\alpha \cdot \eps}\right),$ as otherwise the result is immediate.

    Consider sampling a vector $\vec{p}$ where each $p_i \sim \text{Unif}\left[\frac{0.5}{d}, \frac{1.5}{d}\right].$ Let $\tilde{p} = \vec{p}/(p_1+\cdots+p_d)$. Since $p_1+\cdots+p_d \in [0.5, 1.5],$ this means that $\tilde{p}_i \in [1/(3d), 3/d]$ for all $i$. We define $s$ to be $\sum_{k = 1}^{d} p_i$: note that $s \in [0.5, 1.5]$.
    
    Now, consider the case where each user $i$ is given $\{X_{i, k}\}_{k = 1}^{d}$, where $i$ ranges from $1$ to $n'$ and each $X_{i, k}$ is a nonnegative integer. Then, it is well known that if each $X_{i, k} \sim Pois(m' \cdot p_k)$ (independently), if the user generates $X_{i, k}$ copies of each integer $k$ and randomly permutes them, the resulting sequence has the same distribution as first sampling $m'' \sim Pois(m' \cdot (p_1+\cdots+p_d))$ and then generating $m''$ independent samples each from the discrete distribution $\tilde{p}$. Since $m' \cdot (p_1+\cdots+p_d) \ge m'/2$, with probability at least $1-n' \cdot e^{-\Omega(m')} \ge 1-\frac{\alpha}{10}$, every user $i$ is given at least $m'/4$ samples by a Chernoff bound. Similarly, with probability at least $1-n' \cdot e^{-\Omega(m')} \ge 1-\frac{\alpha}{10}$, every user $i$ is given at most $4m'$ samples by a Chernoff bound, since $m' \cdot (p_1+\cdots+p_d) \le 3m'/2$.
    
    Our overall algorithm $\tilde{M}$ operating on $\{X_{i, k}\}_{i = 1, k = 1}^{n' \hspace{0.25cm} d}$ will work as follows. First, for each user $i$, if $\sum_{k = 1}^{d} X_{i, k} < m'/4$, user $i$ will generate $m$ random samples $\{\tilde{X}_{i, j}\}_{j = 1}^{m}$ from a uniform distribution over $[d]$. Otherwise, user $i$ will generate $X_{i, k}$ copies of the integer $k$ for each $k \in [d]$, randomly permute them, and let $\{\tilde{X}_{i, j}\}_{j = 1}^{m}$ be the first $m \le m'/4$ indices in $i$'s randomly permuted list. Next, the algorithm applies $M$ on $\{\tilde{X}_{i, j}\}$ but just on the first $n$ users (so over $1 \le i \le n, 1 \le j \le m$) to output an estimate for the distribution $\tilde{p} = \frac{\vec{p}}{\sum p_i}$. Finally, $\tilde{M}$ will estimate $p_1+\cdots+p_d$. To do this, let $X_i = \sum_{k = 1}^{d} X_{i, k}$, and let $\tilde{X}_i = \min(\max(X_i, m'/4), 4m')$.
    We apply Theorem \ref{thm:one_dim_estimation} to the values $\tilde{X}_1, \dots, \tilde{X}_{n'}$ with $B = 4m'$, $\tau = \sqrt{m'} \cdot O(\log \frac{d}{\alpha \cdot \eps})$, and $T = O(\log \frac{d}{\alpha \cdot \eps})$, to get an estimate $\tilde{D} = \mathcal{D}(\tilde{X}_1, \dots, \tilde{X}_{n'})$. Letting $\tilde{s} = \tilde{D}/m'$, our final estimate for $\vec{p}$ will be $\tilde{s} \cdot M(\{\tilde{X}_{i, j}\})$.
    
    We show that if $M$ is $(\eps, \delta)$-private, then $\tilde{M}$ is $(2 \eps, \delta)$-private. To see why, note that the generation of samples $\{\tilde{X}_{i, j}\}_{j = 1}^{m}$ from $\{X_{i, k}\}_{k = 1}^{d}$ is entirely user specific, so any algorithm that is private on $\{\tilde{X}_{i, j}\}$ is equally private on $\{X_{i, k}\}$. In addition, $M$ throws out all data after the first $n$ users, which does not worsen privacy. Therefore, the function $M(\{\tilde{X}_{i, j}\}_{i=1, j=1}^{n \hspace{0.38cm} m})$ is still $(\eps, \delta)$-differentially private as a function of $\{X_{i, k}\}$. Next, we similarly just have to check that $\tilde{s}$ is a private function of the $\tilde{X}_i$ values, because each $\tilde{X}_i$ only depends on the $i$th user's input.
    However, Theorem \ref{thm:one_dim_estimation} tells us that in fact $\tilde{s}$ (as a scaled version of $\tilde{D}$) is $(\eps, 0$)-differentially private.
    Therefore, the overall algorithm of generating $\vec{p}$ is $(2 \eps, \delta)$-differentially private.
    
    Next, we verify accuracy of this procedure. Assuming that every user is given $\{X_{i, k}\}_{k = 1}^{d}$ where each $X_{i, k}$ is independently drawn from $\text{Pois}(m' \cdot p_k)$, then $X_i = \sum_{k = 1}^{d} X_{i, k} \sim \text{Pois}(m' \cdot \sum_{k = 1}^{d} p_k).$ In addition, with probability at least $1-n' \cdot e^{-\Omega(m')} \ge 1-(\alpha^{10} \eps^{10})/d^{10}$ (by our assumptions on the sizes of $m', n'$), $X_i \in [m/4, 4m],$ since $p_k \in [0.5/d, 1.5/d]$ for all $k \in [d]$.
    So, $\tilde{X}_i = X_i$ for all $i \in [n']$ with $O(\alpha^9 \eps^9/d^9)$ failure probability (since $n' \le O\left(\frac{d}{\alpha \cdot \eps}\right)$), which means the total variation distance between $X_1+\cdots+X_{n'}$ and $\tilde{X}_1+\cdots+\tilde{X}_{n'}$ is at most $O(\alpha^{9} \eps^{9}/d^{9}).$ Now, note that $\tilde{X}_i \in [0, B]$ for $B = 4m'$, and that $X_i = m' \cdot s \pm O(\sqrt{m'} \cdot \log \frac{d}{\alpha \cdot \eps})$ with probability at least $1-(\alpha^{10} \eps^{10})/d^{10}$ for each $i$ by concentration of Poisson\footnote{See for instance Proposition 11.15 in Goldreich~\cite{GoldreichPropertyTesting}} (and since $0.5 < s < 1.5$). So, for $\tau = C \sqrt{m'} \log \frac{d}{\alpha \cdot \eps}$, we have that with probability at least $1-O(\alpha^9 \eps^9/d^9)$, all of the values $\tilde{X}_i$ for $i \in [n]$ are in the range $[m' \cdot s - \tau, m' \cdot s + \tau]$.
    Overall, this means that if $T = C \log \frac{d}{\alpha \cdot \eps}$, then $\left|\tilde{D}-\frac{X_1+\cdots+X_{n'}}{n'}\right| \ge \frac{8C^2 \sqrt{m'} \log^2 (d/(\alpha \cdot \eps))}{n' \cdot \eps}$ with probability at most $e^{-T}+B/\tau \cdot e^{-n' \cdot \eps/8} = O(\alpha^2 \eps^2/d^2)$, by our bounds on $T, \tau, B,$ and $n'$ and by Theorem \ref{thm:one_dim_estimation}. In addition, $\tilde{D}$ can be always bounded by $O(m')$, so this means that $\tilde{s} = \tilde{D}/m'$ satisfies $\tilde{s} = O(1)$ always and $\left|\tilde{s}-\frac{X_1+\cdots+X_{n'}}{m' \cdot n'}\right| \ge \frac{8C^2 \log^2 (d/(\alpha \cdot \eps))}{\sqrt{m'} \cdot n'}$ with probability at most $O(\alpha^2 \eps^2/d^2)$. Finally, $X_1+\cdots+X_{n'} \sim \text{Pois}(m' \cdot n' \cdot s),$ which means $\BE\left[\left(\frac{X_1+\cdots+X_{n'}}{m' \cdot n'}-s\right)^2\right] = O\left(\frac{1}{m' \cdot n'}\right)$. Overall, as $\tilde{s}$ is always at most $O(1)$ in magnitude, we have
\begin{align*}
    \BE\left[(\tilde{s}-s)^2\right] &\le 2\left(\BE\left[\left(\tilde{s}-\frac{X_1+\cdots+X_{n'}}{m' \cdot n'}\right)^2\right]+\BE\left[\left(\frac{X_1+\cdots+X_{n'}}{m' \cdot n'}-s\right)^2\right]\right) \\
    &\le O\left(\frac{\log^4 \frac{d}{\alpha \cdot \eps}}{m' \cdot (n')^2} + \frac{\alpha^2 \eps^2}{d^2} + \frac{1}{m' \cdot n'}\right).
\end{align*}

    Next, we note that since we are applying the algorithm $M$ to $\{\tilde{X}_{i, j}\}$, which equals $\{X_{i, j}\}$ with probability at least $1-O(\eps^9 \alpha^9/d^9)$. Therefore, by our assumption on the accuracy of $M$, we have that $\BE\left[\|M(\{\tilde{X}_{i, j}\})-\tilde{p}\|_1^2\right] \le \frac{\alpha^2}{1000} + O\left(\frac{\eps^9 \alpha^9}{d^9}\right) \le \frac{\alpha^2}{960}.$
    
    Now, we can bound $\BE\left[\|\tilde{s} \cdot M(\{\tilde{X}_{i, j}\}) - s \cdot \tilde{p}\|_1^2\right].$ (Note that $s \cdot \tilde{p} = \vec{p}$.) First, write $\tilde{s} \cdot M(\{\tilde{X}_{i, j}\}) - s \cdot \tilde{p} = \tilde{s} \cdot \left(M(\{\tilde{X}_{i, j}\})-\tilde{p}\right) + (\tilde{s}-s) \cdot \tilde{p}.$ Also, we may assume that $|\tilde{s}| \le 2$ and we know that $\|\tilde{p}\|_1 = 1$. Therefore,
\begin{align*}
    \BE\left[\|\tilde{s} \cdot M(\{\tilde{X}_{i, j}\}) - s \cdot \tilde{p}\|_1^2\right]
    &\le 2\left(2^2 \cdot \BE\left[\|M(\{\tilde{X}_{i, j}\})-\tilde{p}\|_1^2\right] + \BE\left[(\tilde{s}-s)^2\right]\right) \\
    &\le \frac{\alpha^2}{120} + O\left(\frac{\log^4 \frac{d}{\alpha \cdot \eps}}{m' \cdot (n')^2} + \frac{\alpha^2 \eps^2}{d^2} + \frac{1}{m' \cdot n'}\right).
\end{align*}
    However, by Lemma \ref{lem:similar_to_kamath}, we know that $\BE\left[\|M(X)-\vec{p}\|_2^2\right] \ge \frac{\alpha^2}{24 d}$ if $n' \cdot \sqrt{m'} = o\left(\frac{d}{\alpha \cdot \eps}\right)$. Since the $\ell_1$ norm is at most $\sqrt{d}$ times the $\ell_2$ norm, this means that if $n' \cdot \sqrt{m'} = o\left(\frac{d}{\alpha \cdot \eps}\right)$, then $\BE\left[\|\tilde{s} \cdot M(\{\tilde{X}_{i, j}\})-s \cdot \tilde{p}\|_1^2\right] \ge \frac{\alpha^2}{24}.$ Therefore, assuming the dimension $d$ is at least a sufficiently large constant, and since $n' \ge \log^4 \frac{d}{\alpha \cdot \eps}$, we must have that $\frac{1}{m' \cdot n'} \ge c \cdot \alpha^2$ for some constant $c$. This means that $m \cdot n \le \alpha^{-2}/(c \log^{8} \frac{d}{\alpha \cdot \eps})$. This, however, is impossible, since even the optimal non-private estimation of $\tilde{p}$ (with $n$ users and $m$ samples per user) will have $\ell_2^2$ error at least $\Omega\left(\frac{1}{mn}\right)$ in expectation, and thus $\ell_1^2$ error at least that. However, we claimed that there was a private algorithm (with $n$ users and $m$ samples per user) that could estimate $\tilde{p}$ with error $\BE\left[\|M(X)-\tilde{p}\|_1^2\right] \le \frac{\alpha^2}{1000}$. This is a contradiction, so we must have that in fact $n' \cdot \sqrt{m'} = \Omega\left(\frac{d}{\alpha \cdot \eps}\right)$.
\end{proof}

%

We can now complete the proof of Theorem \ref{thm:lb_main}. The only change we must make is convert a lower bound on the expected error to a lower bound on the error one can achieve with at least $2/3$ probability, which we can do via amplification at the expense of a logarithmic decrease on our lower bound on $n \sqrt{m}$.

\begin{proof}[Proof of Theorem \ref{thm:lb_main}]
    Suppose that with probability at least $2/3$, $\|M(X)-\tilde{p}\|_1 \le c \cdot \alpha$ for some very small constant $c$ and some parameter $\alpha \le 1$. Then, if we have $n \cdot \log \frac{1}{\alpha}$ users, each with $m$ samples, we can obtain $\log \frac{1}{\alpha}$ $(\eps, \delta)$-differentially private estimates of $\tilde{p}$ by splitting the users into $\log \frac{1}{\alpha}$ blocks of size $n$. With probability at least $1-c \cdot \alpha$, at least $\frac{3}{5}$ of the estimates are within $c \cdot \alpha$ in $\ell_1$ distance. In this case, we can take the coordinate-wise median of the estimates to obtain an estimate within $O(c) \cdot \alpha$ of $\tilde{p}$ in $\ell_1$ distance with at least $1-c^2 \cdot \alpha^2$ probability, by a Chernoff bound and by \cite{1centerclustering}. Otherwise, the estimate is still within $1$ of $\tilde{p}$ in $\ell_1$ distance. 

    So, this algorithm, provides an estimate that deviates from $\tilde{p}$ in expected $\ell_1^2$ distance by at most $(O(c) \cdot \alpha)^2 + c^2 \cdot \alpha^2 \le \frac{\alpha^2}{1000}.$ In addition, the algorithm is still $(\eps, \delta)$-differentially private, since the output is a function of $(\eps, \delta)$-differentially private algorithms on disjoint groups of users. By Theorem \ref{thm:lowerbound_remove_poissonization}, this implies that $n \log \frac{1}{\alpha} \cdot \sqrt{m} = \Omega\left(\eps^{-1} \alpha^{-1} d/\log^6\left(\frac{d}{\alpha \cdot \eps}\right)\right)$, so $n \sqrt{m} = \Omega\left(\eps^{-1} \alpha^{-1} d/\log^7\left(\frac{d}{\alpha \cdot \eps}\right)\right)$.
    Finally, even in the non-private case, obtaining error $\alpha$ in $\ell_1$ distance with $2/3$ probability requires $m \cdot n \ge \alpha^{-2} \cdot d$. So, in total, we need that $n \sqrt{m} = \Omega\left(\eps^{-1} \alpha^{-1} d/\log^7\left(\frac{d}{\alpha \cdot \eps}\right)\right)$ and $m \cdot n \ge \alpha^{-2} \cdot d,$ where we are assuming that $\delta = o\left((\eps \cdot \alpha)/(d^2 \cdot \sqrt{m} \cdot \log^2(\frac{d}{\eps \cdot \alpha})\right)$.
    
    Rearranging this, this means we need $\alpha \ge \sqrt{\frac{d}{mn}}$, and $\alpha \cdot \log^7 \left(\frac{d}{\alpha \cdot \eps}\right) \ge \frac{d}{\eps \cdot n\sqrt{m}}.$ If $\alpha^{-1} \ge \frac{d}{\eps}$, then we have that $\alpha \cdot \log^7 \frac{1}{\alpha} \ge \Omega\left(\frac{d}{\eps \cdot n \sqrt{m}}\right),$ so $\alpha \ge \Omega\left(\frac{d}{\eps \cdot n \sqrt{m} \cdot \log^7 (mn)}\right)$. Else, if $\alpha^{-1} < \frac{d}{\eps}$, we have that $\alpha \cdot \log^7 \left(\frac{d}{\eps}\right) \ge \Omega\left(\frac{d}{\eps \cdot n\sqrt{m}}\right),$ so $\alpha \ge \Omega\left(\frac{d}{\eps \cdot n\sqrt{m} \log^7(d/\eps)\cdot}\right)$. So, overall, we have that $\alpha \ge \Omega\left(\sqrt{\frac{d}{mn}} + \frac{d}{\eps \cdot n \sqrt{m} \cdot \log^7 (\eps^{-1} \cdot mnd)}\right)$.
    
    Here, we are assuming that $\delta = o\left((\eps \cdot \alpha)/(d^2 \cdot \sqrt{m} \cdot \log^2(\frac{d}{\eps \cdot \alpha})\right)$. Since we know that $\alpha \ge \Omega\left(\sqrt{d/(mn)}\right),$ this means that as long as $\delta = o\left(\frac{\eps \cdot \sqrt{d/(mn)}}{d^2 \cdot \sqrt{m} \cdot \log^2\left(\eps^{-1} \cdot mnd\right)}\right)$, the result holds. Since the total variation distance is precisely $\frac{1}{2}$ of the $\ell_1$ distance, this concludes the proof.
\end{proof}


\subsection{Proof of Lemma \ref{lem:derivative_commuting}} \label{subsec:complex_analysis}

In this subsection, we prove Lemma \ref{lem:derivative_commuting}, which we omitted in the proof of Theorem \ref{lem:dp_lower_1}. We note that the proof assumes basic knowledge of complex analysis.

\begin{lemma} \label{lem:derivative_commuting}
    Suppose that $X = (X_1, \dots, X_N)$, where each $X_1, \dots, X_N$ is a nonnegative integer. Let $f(X_1, \dots, X_n)$ be a real-valued function that is uniformly bounded in absolute value. Then,
\[\frac{d}{dp}\left[\sum_{X_1, \dots, X_N \ge 0} f(X) \cdot \frac{m^{\sum X_i} p^{\sum X_i}}{X_1! \cdots X_N!}\right] = \sum_{X_1, \dots, X_N \ge 0} \frac{d}{dp}\left[f(X) \cdot \frac{m^{\sum X_i} p^{\sum X_i}}{X_1! \cdots X_N!}\right].\]
\end{lemma}

\begin{proof}
    Assume WLOG that $|f(X)|$ is uniformly bounded by $1$. Then, for any integer $k$,
\[\left|\sum_{\substack{X_1, \dots, X_N \ge 0 \\ X_1+\cdots+X_N = k}}f(X) \cdot \frac{m^{\sum X_i}}{X_1! \cdots X_N!}\right| \le \sum_{\substack{X_1, \dots, X_N \ge 0 \\ X_1+\cdots+X_N = k}} |f(X)| \cdot \frac{m^{\sum X_i}}{X_1! \cdots X_N!} \le (k+1)^N \cdot \frac{m^k}{\lceil k/N \rceil !},\]
    since $|f(X)| \le 1$, $m^{\sum X_i} = m^k,$ $X_1! \cdots X_N! \ge \lceil k/N \rceil!$ if $X_1+\cdots+X_N = k$ since some $X_i$ is at least $\lceil k/N \rceil$, and the number of choices of each $X_i$ is at most $k+1$.
    So, the degree $k$ coefficient of the summation (when the function is thought of as a Taylor series in $p$) is bounded by $(k+1)^N \cdot \frac{m^k}{\lceil k/N \rceil!}$, which as $k \to \infty$ decays faster than any function $\eps^k$ for a fixed $\eps > 0$. Therefore, as a function of $p \in \BC,$ the summation 
\[\sum_{X_1, \dots, X_N \ge 0} f(X) \cdot \frac{m^{\sum X_i} \cdot p^{\sum X_i}}{X_1! \cdots X_N!}\]
    is an entire holomorphic function. This means that its derivative equals the sum of the derivatives of the individual summands because the degree $k$ part of the sum for each $k$ is only a finite sum.
\end{proof}

\subsection{Lower Bound for Estimation with Robustness} \label{subsec:robustness_lower}

We note that Corollary \ref{cor:user_main} does not hold if we also require adversarial robustness, if the number of users $n$ is much more than $\sqrt{d}$. Specifically, we show in this subsection that if even a small constant fraction of users are corrupted, one cannot estimate the mean to better than $\frac{r}{\sqrt{m}}$ with high probability, even if all points are in a \emph{known} ball of radius $r$, and even if we do not require privacy. 
This roughly means that one cannot learn an unknown distribution over $n \gg \sqrt{d}$ users with robustness against a constant fraction of adversarially corrupted users significantly better than one could do with just $\sqrt{d}$ users, both robustly and privately.

\begin{proposition}
    Suppose that $\mathcal{D}$ is an unknown distribution over $\BR^d$, that is promised to have support in the ball of radius $r$ centered at the origin. Suppose that each user $i$ has its $j$th sample $X_{i, j}$ drawn i.i.d. from this distribution $\mathcal{D}$, and then a $O(\kappa)$ fraction of users can have their data corrupted, where $0 < \kappa < 1$ is a constant. Then, there is no algorithm that can learn the mean $\mu$ of $\mathcal{D}$ up to error $o\left(r \cdot \frac{\kappa}{\sqrt{m}}\right)$ with probability better than $2/3$. 
\end{proposition}

\begin{proof}
    First, assume WLOG that $r = 1$ (by scaling). Let $v = (1, 0, \dots, 0) \in \BR^d$, and let $\kappa$ be some unknown constant. Consider the distribution $\mathcal{D}$ that equals $v$ with probability $\frac{1}{2} + \frac{\kappa}{\sqrt{m}}$ and $-v$ with probability $\frac{1}{2} - \frac{\kappa}{\sqrt{m}}$. Likewise, consider the distribution $\mathcal{D}'$ that equals $v$ with probability $\frac{1}{2} - \frac{\kappa}{\sqrt{m}}$ and equals $-v$ with probability $\frac{1}{2} + \frac{\kappa}{\sqrt{m}}$. It is well known that for large values of $m$, the total variation distance between the distribution $Bin(m, \frac{1}{2}+\frac{\kappa}{\sqrt{m}})$ and $Bin(m, \frac{1}{2}-\frac{\kappa}{\sqrt{m}})$ is $O(\kappa)$. Because of this, there exists a method that can take a sequence of $m$ samples drawn from $\mathcal{D}$ and corrupt the distribution with probability at most $O(\kappa)$ (where the choice of corruption is allowed to depend on the samples, but the overall probability of corruption is at most $O(\kappa)$), such that the final sequence is now has the same distribution as $m$ samples drawn from $\mathcal{D}'$.
    
    So, this means that if we do this for all $n$ users, with high probability we do not corrupt more than an $O(\kappa)$ fraction of the users, but now we have completely converted the sequence $\{X_{i, j}\}$ from each sample being drawn from $\mathcal{D}$ to each sample being drawn from $\mathcal{D}'$. Therefore, we cannot distinguish between these two distributions with high probability. But, the difference between the means of these two distributions is precisely $\frac{2 \kappa}{\sqrt{m}}$, so we cannot learn the mean of an unknown distribution up to better than $O\left(\frac{\kappa}{\sqrt{m}}\right)$.
\end{proof}

\section{Using Mean Estimation for Other Learning Problems} \label{sec:erm_and_stuff}

We note that Levy et al.~\cite{levy2021user} showed that user-level private mean estimation could be used in a black-box manner to solve user-level private empirical risk minimization, stochastic convex optimization, and a variant of stochastic gradient descent. In this section, we briefly discuss how our improved mean estimation algorithm in the low user regime improves their algorithms for these three problems.

We recall that the general goal in both empirical risk minimization and stochastic convex optimization is to minimize the average loss. We have a parameter space $\Theta$ and a data space $\mathcal{Z}$, where each user $i$ is given $m$ samples $X_{i, j}$ drawn from some unknown distribution $\mathcal{D}$ over $Z$, as well as a loss function $\ell: \Theta \times \mathcal{Z} \to \BR$. The goal in empirical risk minimization is to minimize the empirical risk $\mathcal{L}(\theta; X) := \frac{1}{mn} \cdot \sum_{i, j} \ell(\theta; X_{i, j})$ over $\theta \in \Theta$ for some data points $X = \{X_{i, j}\}$, whereas the goal in stochastic convex optimization is to minimize the expected population risk $\mathcal{L}(\theta; \mathcal{D}) := \BE_{z \sim \mathcal{D}}[\ell(\theta; z)]$ over $\theta \in \Theta$, assuming each $X_{i, j} \overset{i.i.d.}{\sim} \mathcal{D}$.

We make some assumptions about the loss function and the input space. First, we assume that $\Theta$ is closed, convex, and is contained in a ball of radius $R$, i.e., there exists $\theta_0$ such that $\|\theta-\theta_0\| \le R$ for all $\theta \in \Theta$. Next, we assume that $\ell(\cdot, z)$ is $G$-Lipschitz, meaning that for all $z \in \mathcal{Z}$ and all $\theta \in \Theta$, $\|\nabla \ell(\theta; z)|_2 \le G$, where the gradient is with respect to $\theta$. Next, we assume that the function $\ell(\cdot, z)$ is $H$-smooth, meaning that for all $z \in \mathcal{Z}$ the gradient $\nabla \ell(\theta; z)$ is $H$-Lipschitz as a function of $\theta$. Finally, for any $\theta$, we assume that $\nabla \ell(\theta; Z)$ is $\sigma^2$-subgaussian if $Z \sim \mathcal{D}$, which means that for all $v \in \BR^d$ and $\theta \in \Theta$,
\[\BE_{Z \sim \mathcal{D}}\big[\exp\big(\langle v, \nabla \ell(\theta, Z) - \BE_{Z \sim \mathcal{D}}[\nabla \ell(\theta; Z)]\rangle\big)\big] \le \exp\left(\|v\|_2^2 \cdot \sigma^2/2\right).\]

To perform stochastic gradient descent, we assume that for a function $F: \Theta \to \BR$, we have access to a first-order stochastic oracle $O_{F, \nu^2},$ which is a random mapping such that for all $\theta \in \Theta,$ $\BE\left[O_{F, \nu^2}(\theta)\right] = \nabla F(\theta)$, and $Var\left(O_{F, \nu^2}(\theta)\right) \le \nu^2$. We now consider an abstract optimization algorithm framework using this oracle (in the same way as Levy et al.~\cite{levy2021user}). Given such an oracle $O_{F, \nu^2}$, the optimization algorithm is comprised of three functions, $\textsc{Query}$, $\textsc{Update}$, and $\textsc{Aggregate}$. There also exists an output space $\mathcal{O}$ that, when fed into $\textsc{Query}$, tells us the new value of $\theta \in \Theta$. We start with some initial output $o_0$. Then, for $T$ iterations (i.e., for $t = 0$ to $T-1$), we first compute $\theta_t \leftarrow \textsc{Query}(o_t),$ $g_t = O_{F, \nu^2}(\theta_t)$ (note that $g_t$ is a random output since the oracle is stochastic), and then $o_{t+1} = \textsc{Update}(o_t, g_t)$. Finally, after we have found values $o_0, o_1, \dots, o_T$, we return the final estimate $\widehat{\theta}$ as $\textsc{Aggregate}(o_0, o_1, \dots, o_T)$.

We note that the algorithm's guarantees will depend on the parameters $R, G, \sigma, \nu,$ as well as on the number of iterations $T$ of the stochastic gradient descent. In certain situations, the guarantees may also depend on $\mu$, a guarantee on strong convexity. Namely, we may assume that $\ell(\cdot, z)$ is $\mu$-\emph{strongly convex}, which means that for any $\theta_1, \theta_2 \in \Theta$ and any $z \in \mathcal{Z}$, $\langle \nabla f(\theta_1, z) - \nabla f(\theta_2, z), \theta_1-\theta_2 \rangle \ge \mu \cdot \|\theta_1-\theta_2\|^2$. Finally, the algorithm's guarantees depend on $n$, the number of users, $m$, the number of samples per user, and $d$, the dimension.

We first state the following proposition about stochastic gradient methods, due to results from \cite{bubeck2014convexoptimization, davis2019convex, kulunchakov2020convex} and compiled by \cite{levy2021user}.

\begin{proposition} \cite{levy2021user} \label{prop:projected_SGD}
Let $F: \Theta \to \BR$ be an $H$-smooth function, and assume we have oracle access to $O_{F, \nu^2}$. Then, there exists an optimization algorithm, comprised of three functions $\textsc{Query}$, $\textsc{Update}$, and $\textsc{Aggregate}$, with output $\widehat{\theta}$, with the following convergence guarantees:
\begin{enumerate}
    \item \cite{bubeck2014convexoptimization} If $F$ is a convex function, then 
\[\BE\left[F(\widehat{\theta})-\inf_{\theta \in \Theta} F(\theta)\right] \le O\left(\frac{H \cdot R^2}{T} + \frac{\nu \cdot R}{\sqrt{T}}\right).\]
    \item \cite{kulunchakov2020convex} If $F$ is a $\mu$-strongly convex function, and we have access to $\theta_0 \in \Theta$ such that $F(\theta_0)-\inf_{\theta \in \Theta} f(\theta) \le \Delta_0$ for some $\Delta_0 \ge 0$, then 
\[\BE\left[F(\widehat{\theta}) - \inf_{\theta \in \Theta} F(\theta)\right] \le O\left(\Delta_0 \cdot \exp\left(-\frac{\mu}{H} \cdot T\right) + \frac{\nu^2}{\mu \cdot T}\right).\]
    \item \cite{davis2019convex} For a point $x$ in the Euclidean space that contains $\Theta$, define $\Pi_{\Theta}(x) = \arg\min_{\theta \in \Theta} \|x-\theta\|_2$. In addition, for $\gamma > 0$, define the gradient mapping $G_{F, \gamma}$ as
\begin{equation} \label{eq:G}
    G_{F, \gamma}(\theta) := \frac{1}{\gamma} \cdot \left[\theta-\Pi_{\Theta}(\theta-\gamma \nabla F(\theta))\right].
\end{equation}
    Assuming that we have access to $\theta_0 \in \Theta$ such that $\|G_{F, 1/H}(\theta_0)\|_2 - \inf_{\theta \in \Theta} \|G_{F, 1/H}(\theta)\|_2 \le \Delta_1$ for some $\Delta_1 \ge 0$, then
\[\BE\|G_{F, 1/H}(\widehat{\theta})\|_2^2 \le O\left(\frac{H \cdot \Delta}{T} + \nu \cdot \sqrt{\frac{H \cdot \Delta_1}{T}}\right).\]
\end{enumerate}
\end{proposition}

Recall our assumptions on $\Theta$ and the loss function $\ell$. Also, recall the definition of $\mathcal{L}(\theta, X)$ for $X = \{X_{i, j}\}$.
We are now ready to state the main guarantee of Levy et al.~\cite{levy2021user} regarding private user-level empirical risk minimization.  We then describe how our mean estimation algorithm can be used to improve it.

\begin{theorem} \cite{levy2021user} \label{thm:erm}
Suppose that $n = \tilde{\Omega}(\sqrt{d \cdot T}/\eps)$, where the $\tilde{\Omega}$ hides logarithmic factors in $\frac{1}{\delta}, T, d, m, \frac{B}{\sigma}, \frac{1}{\alpha}$, where $\alpha$ is the failure probability and $B$ is a ball that is promised to contain all data points $X_{i, j}$ even in the worst case. Then, there exists an algorithm (Algorithm 4 in Levy et al.~\cite{levy2021user}) that takes as input $X = \{X_{i, j}\}$ and is instantiated with a first-order gradient algorithm, and outputs some $\widehat{\theta} \in \Theta$, with the following properties. First, the algorithm is $(\eps, \delta)$ user-level differentially private for any $0 < \eps, \delta \le 1.$ In addition, there exist variants of projected stochastic gradient descent (the algorithm with guarantees described in Proposition \ref{prop:projected_SGD}) such that if Algorithm 4 in Levy et al.~\cite{levy2021user} is instantiated with the projected SGD, then:
\begin{enumerate}
    \item If for all $z \in \mathcal{Z}$, $\ell(\cdot, z)$ is convex, then 
\[\BE\left[\mathcal{L}(\widehat{\theta}; X) - \inf\limits_{\theta \in \Theta} \mathcal{L}(\theta, X)\right] = \tilde{O}\left(\frac{R^2 H}{T} + R \sigma \cdot \frac{d}{n \sqrt{m} \cdot \eps}\right).\]
    \item If for all $z \in \mathcal{Z}$, $\ell(\cdot, z)$ is $\mu$-strongly convex, then 
\[\BE\left[\mathcal{L}(\widehat{\theta}; X) - \inf\limits_{\theta \in \Theta} \mathcal{L}(\theta, X)\right] = \tilde{O}\left(G R \cdot \exp\left(-\frac{\mu}{H} \cdot T\right) + \sigma^2 \cdot \frac{d^2}{\mu \cdot n^2 m \cdot \eps^2}\right).\]
    \item Defining $G_{F, \gamma}(\theta)$ as in Equation \eqref{eq:G}, we have
\[\BE\left[\|G_{\mathcal{L}(\cdot, X), 1/H}(\widehat{\theta})^2\|_2^2\right] = \tilde{O}\left(\frac{H^2 R}{T} + H R \sigma \cdot \frac{d}{n\sqrt{m} \cdot \eps}\right).\]
\end{enumerate}
\end{theorem}

The algorithm of Levy et al.~\cite{levy2021user} importantly uses the high-dimensional mean estimation $T$ times, and so by strong composition requires $(\eps', \delta')$-differential privacy for each iteration, where $\eps' = \Theta(\eps/\sqrt{T})$ and $\delta' = \Theta(\delta/T)$. Because of this, they require $n = \tilde{\Omega}(\sqrt{d}/\eps')$, as they are only able to get mean estimation algorithms in this regime. However, because of our extension of user-level private mean estimation to $n = \tilde{\Omega}(1/\eps' \cdot \log (1/\delta'))$, their theorem can be immediately improved as follows.

\begin{theorem} \label{thm:erm_improved}
    Theorem \ref{thm:erm} still holds even if $n = \tilde{\Omega}(\sqrt{T}/\eps)$, where $\tilde{\Omega}$ hides logarithmic factors in $\frac{1}{\delta}, T, d, m, \frac{B}{\sigma}, \frac{1}{\alpha}$.
\end{theorem}

Hence, we provide an extension of user-level private empirical risk minimization, which uses a variant of projected stochastic gradient descent, to the regime where $n$ is relatively small, i.e., $n \ge \tilde{\Omega}(\sqrt{T}/\eps)$ (which has almost no dependence on $d$) as opposed to $n \ge \tilde{\Omega}(\sqrt{dT}/\eps)$.

Next, we describe how we also can get improvements for the stochastic convex optimization problem. Here, instead of trying to minimize $\mathcal{L}(\widehat{\theta}, X)$, we wish to minimize $\mathcal{L}(\widehat{\theta}, \mathcal{D})$, i.e., minimize the expected loss over the full unknown distribution rather than just over the known samples. Recall our assumptions on $\Theta$ and the loss function $\ell$, and the distribution $\mathcal{D}$ over the data space $\mathcal{Z}$. We are now ready to describe the stochastic convex optimization guarantees of Levy et al.~\cite{levy2021user}. We then describe how our mean estimation algorithm can be used to improve it.

\begin{theorem} \label{thm:sco}
    Define $\underline{G} = \min(G, \sigma \sqrt{d})$, and suppose that $n = \tilde{\Omega}(\min(\sqrt[3]{d^2 m H^2 R^2/(G \underline{G} \cdot \eps^4)}, HR\sqrt{m}/(\sigma \eps))$. Then, there exists a stochastic convex optimization algorithm (Algorithm 5 in Levy et al.~\cite{levy2021user}) that takes as input $X = \{X_{i, j}\}$ and has the following properties. The algorithm is $(\eps, \delta)$ user-level differentially private, and if $\mathcal{D}, \ell$ satisfy our assumptions and each $X_{i, j} \in \mathcal{Z}$ is drawn i.i.d. from $\mathcal{D}$, then the output $\widehat{\theta}$ satisfies
\[\BE\left[\mathcal{L}(\widehat{\theta}, \mathcal{D})\right] - \min_{\theta \in \Theta} \mathcal{L}(\theta, \mathcal{D}) = \tilde{O}\left(\frac{R \cdot \sqrt{G \cdot \underline{G}}}{\sqrt{mn}} + R \sigma \cdot \frac{d}{n\sqrt{m} \cdot \eps}\right).\]
\end{theorem}

Indeed, their proof at one point assumes that $n \ge \tilde{\Omega}(\eps \cdot \sqrt{d \cdot H/\lambda})$, where they will set $T = \tilde{\Theta}(H/\lambda)$ as that was the necessary assumption from Theorem \ref{thm:erm}, and where $\lambda = \sqrt{\frac{\sigma^2 \cdot d^2}{n^2 m \eps^2} + \frac{G \cdot \underline{G}}{n \cdot m}}/R$. (We allow the $\tilde{\Omega}$ to hide all logarithmic factors.) However, because of our improved Theorem \ref{thm:erm_improved}, we will only need that $n \ge \tilde{\Omega}\left(\eps \cdot \max(1, \sqrt{H/\lambda})\right)$ instead. Hence, we obtain the following theorem.

\begin{theorem}
    Theorem \ref{thm:sco} holds even if $n = \tilde{\Omega}\left(\eps^{-1} + \min\left(\sqrt[3]{m H^2 R^2/(G \underline{G} \cdot \eps^4)}, HR \sqrt{m}/(d \cdot \sigma \cdot \sqrt{\eps})\right)\right)$, where $\tilde{\Omega}$ hides all logarithmic factors.
\end{theorem}

Hence, we provide an extension of user-level private stochastic convex optimization to a smaller regime of $n$, again having almost no dependence on $d$.


\end{document}